\documentclass{article}

\usepackage[utf8]{inputenc}
\usepackage[english]{babel}
\usepackage[T1]{fontenc}
\usepackage{graphicx} % Required for inserting images
\usepackage{algorithm}
\usepackage{fullpage}
\usepackage{amsmath, amssymb, amsthm}
\usepackage{amsfonts}
\usepackage{xspace}
\usepackage{color}
\usepackage{xcolor}
\usepackage{booktabs}
\usepackage{thmtools}
\usepackage{mathtools}
\usepackage{xspace}
\usepackage{thm-restate}
\definecolor{darkgreen}{rgb}{0,0.5,0}
\usepackage{hyperref}
\hypersetup{
    unicode=false,          % non-Latin characters in Acrobat’s bookmarks
    colorlinks=true,        % false: boxed links; true: colored links
    linkcolor=blue,          % color of internal links (change box color with linkbordercolor)
    citecolor=purple,        % color of links to bibliography
    filecolor=magenta,      % color of file links
    urlcolor=cyan           % color of external links
}

\usepackage[noend]{algpseudocode}
\usepackage[colorinlistoftodos,textsize=small,color=red!25!white,obeyFinal]{todonotes}
\usepackage{soul}
\usepackage{enumerate}

\usepackage[capitalize, nameinlink]{cleveref}
\usepackage[section]{placeins}
\crefname{theorem}{Theorem}{Theorems}
\Crefname{lemma}{Lemma}{Lemmas}
\Crefname{invariant}{Invariant}{Invariants}
\Crefname{claim}{Claim}{Claims}
\Crefname{observation}{Observation}{Observations}
\Crefname{algorithm}{Algorithm}{Algorithms}
\Crefname{figure}{Figure}{Figures}
\Crefname{challenge}{Challenge}{Challenges}

\newtheorem{theorem}{Theorem}[section]
\newtheorem{lemma}[theorem]{Lemma}

\newtheorem{corollary}[theorem]{Corollary}

\newtheorem{claim}[theorem]{Claim}
\newtheorem{observation}[theorem]{Observation}
 \newtheorem{result}{Result}

\title{Parallel Set Cover and Hypergraph Matching\\ via Uniform Random Sampling
}

\author{
    Laxman Dhulipala \\ Google Research
    \and Michael Dinitz \\ Johns Hopkins University
    \and Jakub Łącki \\ Google Research, New York
    \and Slobodan Mitrovi\' c \\ UC Davis
}

\newcommand{\eqdef}{\stackrel{\text{\tiny\rm def}}{=}}
\newcommand{\sample}[0]{\textsc{Sample}}
\newcommand{\ssp}[0]{\textsc{SSP}\xspace}
\newcommand{\setcover}[0]{\textsc{SetCover}\xspace}
\newcommand{\matching}[0]{\textsc{Matching}\xspace}
\newcommand{\hm}[0]{\textsc{HypergraphMatching}\xspace}
\newcommand{\vcover}[0]{\textsc{VertexCover}\xspace}
\newcommand{\poly}[0]{\mathrm{poly}}
\newcommand{\E}[1]{\mathbb{E}[#1]}
\newcommand{\prob}[1]{\Pr \left[ #1 \right]}

\newcommand{\ceil}[1]{\left \lceil #1 \right \rceil}
\newcommand{\floor}[1]{\left \lfloor #1 \right \rfloor}
\newcommand{\rb}[1]{\left( #1 \right)}
\newcommand{\cb}[1]{\left\{ #1 \right\}}
\newcommand{\eps}{\epsilon}
\newcommand{\bD}{\bar{D}}
\newcommand{\bh}{\bar{h}}
\newcommand{\tN}{\tilde{N}}
\newcommand{\tO}{\tilde{O}}

\newcommand{\manis}[0]{\textsc{MaNIS}\xspace}

\newcommand{\hdelta}[0]{\ensuremath{H_{\Delta}}}
\newcommand{\ho}{\hat{O}}

\usepackage{xcolor}
\iftrue

\newcommand{\mdnoteinline}[1]{\todo[inline, size=\normalsize, color=green!40]{Mike's Note: #1}}
\newcommand{\stodo}[1]{\todo[color=purple!10,inline]{Slobodan: #1}}
\newcommand{\kuba}[1]{{\color{blue} Kuba: #1}}
\newcommand{\slobo}[1]{{\color{purple} Slobodan: #1}}
\newcommand{\laxman}[1]{{\color{orange} Laxman: #1}}
\newcommand{\lnote}[1]{\todo[inline, size=\normalsize, color=brown!40]{Laxman: #1}}

\fi

\begin{document}
%\clearpage

\maketitle
\begin{abstract}
The \textsc{SetCover} problem has been extensively studied in many different models of computation, including parallel and distributed settings.
From an approximation point of view, there are two standard guarantees: an $O(\log \Delta)$-approximation (where $\Delta$ is the maximum set size) and an $O(f)$-approximation (where $f$ is the maximum number of sets containing any given element).

In this paper, we introduce a new, surprisingly simple, model-independent approach to solving \textsc{SetCover} in unweighted graphs. We obtain multiple improved algorithms in the MPC and CRCW PRAM models.
First, in the MPC model with sublinear space per machine, our algorithms can compute an $O(f)$ approximation to \textsc{SetCover} in $\hat{O}(\sqrt{\log \Delta} + \log f)$ rounds\footnote{We use the $\hat{O}(x)$ notation to suppress $\mathrm{poly} \log x$ and $\mathrm{poly} \log \log n$ terms.} and a $O(\log \Delta)$ approximation in $O(\log^{3/2} n)$ rounds.
Moreover, in the PRAM model, we give a $O(f)$ approximate algorithm using linear work and $O(\log n)$ depth.
All these bounds improve the existing round complexity/depth bounds by a $\log^{\Omega(1)} n$ factor.

Moreover, our approach leads to many other new algorithms, including improved algorithms for the \textsc{HypergraphMatching} problem in the MPC model, as well as simpler \textsc{SetCover} algorithms that match the existing bounds.
\end{abstract}

\thispagestyle{empty}

\newpage

\tableofcontents

\thispagestyle{empty}

\newpage

% we introduce simple, model-independent
% PRAM O(log^2 n) set cover, PRAM O(log n) set cover
% f-approximate set cover in O(sqrt(log n))
% MPC hyper-matching

\clearpage\pagenumbering{arabic}

\section{Introduction}

There is perhaps no more central and important problem in the area of approximation algorithms than \setcover.
It has been a testbed for various algorithmic techniques that have become central in the field: greedy algorithms, deterministic and randomized rounding, primal-dual, dual fitting, etc.  Due to its importance, ubiquity, and the fact that many different algorithmic techniques can be used, it is widely considered a ``textbook problem'' and, for example, has been used to illustrate the very basics of approximation algorithms~\cite[Chapter 1]{WS11}.  There are essentially two standard approximation bounds, both of which can be achieved through a number of different algorithms: an $f$-approximation, where $f$ is the \emph{frequency} (the maximum number of sets containing any given element), and an $H_{\Delta} = O(\log \Delta)$-approximation, where $\Delta$ is the maximum set size and $H_k$ is the $k$'th harmonic number.\footnote{This result is often given as an $O(\log n)$-approximation since $\Delta \leq n$, but $H_{\Delta}$ is a technically stronger bound.}

Unsurprisingly, \setcover has also received significant attention in parallel and distributed models of computation.  However, the simple sequential algorithms for \setcover are not ``obviously'' parallelizable, so new algorithms have been developed for these models.  These lines of work range from classical complexity-theoretic models (e.g., showing that it can be approximated well in $\mathsf{NC}$~\cite{berger1994efficient}), classical parallel models such as PRAMs~\cite{berger1994efficient,rajagopalan1998primal,blelloch2011linear}, classical distributed models such as LOCAL~\cite{kuhn2006price,koufogiannakis2011distributed}, and modern models such as MapReduce and Massively Parallel Computation (MPC)~\cite{stergiou2015set,bateni2018optimal}. Much of this work has been model-focused rather than model-independent, and ideas and techniques from one model can only sometimes be transferred to a different model.

In this paper, we introduce a new, simple, and model-independent technique for solving unweighted \setcover in parallel settings. 
Our technique, which involves careful independent random sampling of either the sets or elements, yields both a $(1+\epsilon)f$-approximation and a $(1+\epsilon) \hdelta$-approximation and can be efficiently instantiated in multiple models of computation, including the MPC and PRAM models.
Moreover, it can also be extended to solve the approximate \hm problem in unweighted graphs.
By applying our technique, we obtain efficient algorithms for \setcover and \hm in MPC and PRAM models, which either improve upon or (essentially) match state-of-the-art algorithms for the problems.
Importantly, our technique provides a unified and model-independent approach across \hm and two variants of \setcover, and can be efficiently implemented in two fundamental models of parallel computation.

Our algorithms are obtained by parallelizing two classic $f$- and $O(\log \Delta)$-approximate \setcover algorithms.
The $f$-approximate algorithm repeatedly picks an uncovered element and adds all sets containing it to the solution.
The $O(\log \Delta)$-approximate in each step simply adds to the solution the set that covers the largest number of uncovered points.

Even though our parallelization of these algorithms is surprisingly direct, to the best of our knowledge, it has not been analyzed prior to our work.
At a high level, our algorithms perform independent random sampling to find a collection of sets to be added to the solution, remove all covered elements from the instance, and then repeat. By combining the random sampling-based approach with modern techniques in parallel algorithms, we are able to give state-of-the-art bounds.

\subsection{Our Contribution}
We now present the main contributions of the paper.
We study the unweighted version of \setcover.
To formulate the bounds we obtain, we assume the \setcover problem is represented by a bipartite graph, in which vertices on one side represent the sets, and vertices on the other side represent elements to be covered.
Edges connect elements with all sets that they belong to.
We use $\Delta$ to denote the maximum degree of a vertex representing a set, and $f$ to denote the maximum degree of a vertex representing an element.
We use $n$ to denote the number of vertices in the graph (equal to the number of sets plus the number of elements) and $m$ to denote the number of edges (the total size of all sets).

We start by presenting our results in the 
Massively Parallel Computation (MPC) model~\cite{karloff2010model,goodrich2011sorting,beame2017communication,andoni2014parallel}.
MPC computation proceeds in \emph{synchronous} rounds over $M$ machines. 
We assume that the input to the computation is partitioned arbitrarily across all machines in the first round.
Each machine has a local space of $\eta$ bits. In one round of computation, a machine first performs computation on its local data. Then, the machines can communicate by sending messages: each machine can send messages to any other machine.
The messages sent in one round are delivered at the beginning of the next round.
Hence, within a round, the machines, given the messages received in this round, work entirely independently.
Importantly, the total size of the messages sent or received by a machine in a given round is at most $\eta$ bits.

In the context of graph algorithms, there are three main regimes of MPC defined with respect to the relation of the available space on each machine $\eta$ to the number of vertices of the graph $n$.
In the \emph{super-linear} regime, $\eta = n^{1+c}$ for a constant $0 < c < 1$.
The \emph{nearly-linear} regime requires  $\eta = n \;\poly \log n$. Finally, the most restrictive and challenging \emph{sub-linear} regime requires $\eta = n^c$.
In all the regimes, we require that the total space of all machines is only a $\poly \log n$ factor larger than what is required to store the input.
% slightly larger than what is required to store the input.
% Specifically, if the input size is $N$, we require $\eta \cdot M = N^{1+c}$ for some constant $c > 0$, which can be made arbitrarily small.

In our definition of the \setcover problem, the number of vertices is the number of sets plus the number of elements.
We note that some \setcover algorithms in the linear space regime (both prior and ours) only require space near-linear in the number of sets plus sublinear in the number of elements, but we use a single parameter for simplicity.

\begin{table*}[t!]
%\footnotesize
    \centering
%    \setlength\tabcolsep{3pt}
%    \footnotesize
    \begin{tabular}{*6c}
        \toprule
        Ref. & Space/Machine & Total Space & Approx. Factor & Det. & Round Complexity \\
        \midrule
        \cite{blelloch2011linear} &  $O(n^{\delta})$ & $O(m)$ & $(1+\epsilon)\hdelta$ & No & $O(\log^2 n)^{\ddag}$ \\
        Here &  $O(n^{\delta})$ & $\tilde{O}(m)$ & ${(1+\epsilon)\hdelta}^{\ddag}$ & No & $\ho(\log \Delta \cdot \sqrt{\log f} )^{\ddag}$ \\
        Here & $O(n^{\delta})$ &  $\tilde{O}(m)$ & $(1+\epsilon)f^{\ddag}$ & No & $\ho(\sqrt{\log \Delta} )^{\ddag}$\\
        ~\cite{ben2019optimal} & $O(n^{\delta})$ & $O(m)$ & $f + \epsilon$ & Yes & $O(\log \Delta / \log\log \Delta)$ \\
        ~\cite{even2018distributed} & $O(n^{\delta})$ & $O(m)$ & $(1+\epsilon)f$ & Yes & $O(\log (f\Delta) / \log\log (f\Delta))$ \\
        \midrule
        \cite{bateni2018optimal} & $\tilde{O}(n)$ & $\tilde{O}(m)$  & $O(\log n)^{\ddag}$ & No & $O(\log n)^{\ddag}$ \\
        Here & $\tilde{O}(n)$ & $\tilde{O}(m)$ & ${(1+\epsilon)\hdelta}^{\ddag}$ & No & $O(\log \Delta)^{\ddag}$ \\
        \midrule
        ~\cite{harvey2018greedy} & $O(f n^{1+c})$ & $O(m)$ & $f$ & No & $O((\phi/ c)^2)$ \\
%        \cite{ben2019optimal,even2018distributed}
%        \cite{kuhn2006price} & $O(\log n)$ & $O(\log n)^{\ddag}$ & Local & \footnotesize{Uses LS Decomp.} \\
%        \cite{koufogiannakis2011distributed} & $f$ & $O(\log n)^{\ddag}$  & Local & \footnotesize{Uses LS Decomp.} \\
        \bottomrule
    \end{tabular}
    \caption{Round complexity of \setcover algorithms in the Massively Parallel Computation model.
    We use $n$ to denote the number of vertices in the graph (which is equal to the number of sets plus the number of elements) and $m$ to denote the number of edges (the total size of all sets).
    $\delta \in (0, 1)$ is a constant.
%    $m$ denotes the sum of the sizes of all sets (or the number of edges in the bipartite
%    representation of \setcover), and $n$ denotes the number of elements. 
    In~\cite{harvey2018greedy} $\phi \in (0, 1]$ is any value satisfying $m \leq n^{1+\phi}$
    and $c < \phi$ controls the amount of space per machine.
    We use $^{\ddag}$ to denote
    that a bound holds with high probability.
    }\label{table:relworkpram-2}
\end{table*}

\paragraph*{\setcover in MPC}
Our first result is a set of improved MPC algorithms for \setcover.  

\begin{result}[\cref{thm:MPC-log-approx-linear,thm:log-approx-MPC-sub-linear,thm:f-approx-MPC} rephrased]
\label{result:MPC-setcover}
Let $\eps \in (0, 1/2)$ be a constant. Denote by $f$ the maximum number of sets an element appears in and by $\Delta$ the largest set size. 
Then, \setcover can be solved in MPC with the following guarantees: 
\begin{itemize}
  \item $(1+\eps) \hdelta$-approximation~in $\ho\left(\poly(1/\epsilon) \cdot \log \Delta \cdot \sqrt{\log f} \right)$ rounds in the sub-linear regime,
    \item $(1+\eps) \hdelta$-approximation~in $O(\log \Delta)$ rounds in the nearly-linear regime,
    \item $(1+\eps) f$-approximation~in $\ho\left(\poly(1/\epsilon) \cdot \rb{\sqrt{\log \Delta} + \log f} \right)$ rounds in the sub-linear regime.
\end{itemize}
%The first and third algorithms have $\poly(1/\epsilon)$ dependence in the round complexity.
The algorithms use $\tilde{O}(m)$ total space, and the round complexities hold with high probability.
\end{result}

Before our work, the best-known round complexity for the $(1+\eps) H_{\Delta}$ \setcover in the sub-linear regime was $O(\log \Delta \cdot \log f)$; this complexity is implicit in \cite{berger1994efficient}. 
Our algorithm improves this bound by a $\sqrt{\log f}$ factor.
In the nearly-linear space regime, it is possible to achieve $O(\log n)$-approximation in $O(\log n)$ rounds by building on \cite{bateni2018optimal}. It is unclear how to transfer this approach to the sub-linear regime.
We improve the approximation ratio to $\hdelta$, which is better, especially when $\Delta \ll n$.

In terms of $(1+\eps) f$-approximation, the most efficient \setcover algorithm in MPC follows by essentially a direct adaption of the \textsc{Congest}/\textsc{Local} $O(\log \Delta / \log \log \Delta)$ round algorithms in \cite{ben2019optimal,even2018distributed} to MPC.
Hence, for $f \le 2^{O(\sqrt{\log n})}$, our work improves the MPC round complexity nearly quadratically.

\paragraph*{\setcover in PRAM}
Since our main algorithmic ideas are model-independent, they also readily translate to the PRAM setting, giving a new result for $(1+\epsilon)f$-approximate \setcover that improves over the state-of-the-art, and a streamlined $(1+\epsilon)\hdelta$-approximation algorithm for \setcover~\cite{blelloch2011linear}.

\begin{table*}[t!]
%\footnotesize
    \centering
%    \setlength\tabcolsep{3pt}
%    \footnotesize
    \begin{tabular}{*6c}
        \toprule
        Ref. & Approx. Factor & Det. & Work & Depth & Notes \\
        \midrule
%        \cite{berger1994efficient} & $(1+\epsilon)H_{\Delta}$ & $O(m \log^{5} n)$ & $O(\log^5 n)$ & \footnotesize{Uses independent sampling.}\\
%        \cite{rajagopalan1998primal} & $2(1+\epsilon)H_{\Delta}$ & $O(m \log^{3} m)^{\ddag}$ & $O(\log^{3}(mn))^{\ddag}$ & \footnotesize{Introduces permutation sampling.}\\
%        \cite{blelloch2011linear} & $(1+\epsilon)H_{\Delta}$ & $O(m)^{*}$ & $O(\log^3 n)^{\ddag}$ & \footnotesize{Work-efficient $H_{\Delta}$-approximation.}\\
%        (This Paper) & ${(1+\epsilon)H_{\Delta}}^{*}$ & $O(m)$ & $O(\log^2 n \log\log n)$ &  \\
%        \midrule 
        \cite{khuller1994primal} & $(1+\epsilon)f$ & Yes & $O(f m)$ & $O(f \log^{2} n)$ &  \\
        \cite{koufogiannakis2011distributed} & $2$ & No & $O(m)^{*}$ & $O(\log n)^{\ddag}$ & \footnotesize{For weighted instances with $f=2$.}\\
        Here & $(1+\epsilon)f^{*}$ & No & $O(m)$ & $O(\log n)$ &  \\
        \bottomrule
    \end{tabular}
    \caption{Parallel cost bounds (work and depth) of $f$-approximate \setcover algorithms in the CRCW PRAM.
    $m$ denotes the sum of the sizes of all sets (or the number of edges in the bipartite
    representation of \setcover), $n$ denotes the number of elements, $f$ denotes the maximum number of sets any element is contained in, and $\epsilon \in (0, 1/2)$ is an arbitrary constant.
    We use $^{*}$ to denote that a bound holds in expectation, and $^{\ddag}$ to denote
    that a bound holds with high probability.
    }\label{table:relworkpram}
\end{table*}

\begin{result}[\cref{thm:f-approx-PRAM,thm:PRAM-log-approx-setcover} rephrased]
\label{result:PRAM-setcover}
Let $\epsilon \in (0, 1/2)$ be an absolute constant.
Let $f$ be the maximum number of sets an element appears in, and let $\Delta$ be the largest set size. Then, \setcover can be solved in CRCW PRAM with the following guarantees:
\begin{itemize}
    \item $(1+\epsilon)f$-approx.~in expectation with \emph{deterministic}  $O(n + m)$ work and $O(\log n)$ depth.
    \item $(1+\epsilon)\hdelta$-approx.~in expectation with \emph{deterministic} $O(n + m)$ work and $O(\log^2 n \log\log n)$ depth.
\end{itemize}
\end{result}

In the context of $(1+\eps)f$-approximation, our result improves the state-of-the-art~\cite{khuller1994primal} total work by $f$ while depth is improved by an $f \log n$ factor.
For $(1+\epsilon)\hdelta$-approximation, 
our result obtaining deterministic $O(\log^2 n \log\log n)$ depth and providing the approximation guarantee in expectation should be compared to the state-of-the-art PRAM algorithm of Blelloch, Peng, and Tangwongsan~\cite{blelloch2011linear}, which provides a depth guarantee in expectation and a worst-case guarantee for the approximation ratio.
While the expected depth bound reported in \cite{blelloch2011linear} is $O(\log^3 n)$, we believe it can be improved to $O(\log^2 \log \log n)$ using some of the implementation ideas in our PRAM algorithm (see the full version for more discussion).
A more detailed comparison between the prior and our results in PRAM is given in \cref{table:relworkpram}.\footnote{Our algorithms provide approximation in expectation. Nevertheless, this can be lifted to ``with high probability'' guarantees by executing $O(\log n / \eps)$ independent instances of our algorithm and using the smallest set cover. It incurs an extra $O(\log n / \eps)$ factor in the total work while not affecting the depth asymptotically.}

\paragraph*{\hm in MPC}
Finally, we also obtain an improved MPC algorithm for finding hypergraph matchings, i.e., for finding matchings in graphs where an edge is incident to (at most) $h$ vertices.

\begin{result}[\cref{thm:hypermatching-MPC} rephrased]
\label{result:hypermatching-MPC}
    Let $\eps \in (0, 1/2)$ be an absolute constant. There is an MPC algorithm that, in expectation, computes a $(1-\eps) / h$ approximate maximum matching in a rank $h$ hypergraph in the sub-linear space regime. 
    This algorithm succeeds with high probability, runs in $\ho\rb{\poly(1/\epsilon) \cdot \rb{h^4 + h\cdot \sqrt{\log \Delta}}}$ MPC rounds and uses a total space of $\tO\rb{m}$. 
\end{result}
Prior work~\cite{hanguir2021distributed} shows how to solve \hm in rank $h$ hypergraphs in $O(\log n)$ rounds in the \emph{nearly-linear} space regime.
So, for $h \in O(1)$, \cref{result:hypermatching-MPC} improves quadratically over the known upper bound and, in addition, extends to the sub-linear space regime at the cost of slightly worsening the approximation ratio. For simple graphs, i.e., when $h = 2$, the work~\cite{ghaffari2019sparsifying} already provides $\tO(\sqrt{\log \Delta})$ round complexity algorithm for computing $\Theta(1)$-approximate, and also maximal, matching. 
Nevertheless, our approach is arguably simpler than the one in \cite{ghaffari2019sparsifying} and, as such, lands gracefully into the MPC world.

\subsection{Further Related Work} \label{sec:related}

%\setcover{} is one of the canonical problems in approximation algorithms leading to the development of many interesting techniques in both sequential and parallel approximation algorithms.
%In this section, we review parallel approximate set cover algorithms in the PRAM and MPC models and describe how our new results relate to this prior work.
%We also discuss related work on parallel algorithms for matchings and hyper-matchings.

%The goal in parallel approximate \setcover algorithms in shared-space models like the PRAM is to obtain good approximation guarantees while ensuring work-efficiency (asymptotically requiring the same work or total number of operations as the fastest sequential algorithm) and poly-logarithmic depth (longest chain of sequential dependencies).

%\kuba{The next section should come later.}
%\kuba{How much should we say about log n approximate \setcover in PRAM?}

\paragraph*{\setcover in the MPC Model}
Both  \setcover and \vcover, i.e., \setcover with $f=2$, have been extensively studied in the MPC model.
Stergiou and Tsioutsiouliklis~\cite{stergiou2015set} studied the \setcover problem in MapReduce and provided an empirical evaluation.
Their main algorithm is based on bucketing sets to within a $(1+\epsilon)$ factor with respect to the set sizes and then processing all the sets within the same bucket on one machine.
Their algorithm, when translated to the MPC model, runs in $O(\log \Delta)$  iterations, but does not come with a bound on the required space per machine, which in the worst case can be linear in the input size.
%\kuba{Do they provide a bound on the number of rounds?}
%\laxman{Added round-complexity}

Harvey, Liaw, and Liu~\cite{harvey2018greedy} studied weighted \vcover and \setcover in the MPC model and obtained results for both $f$ and $(1+\epsilon)\hdelta$-approximation. 
Their results exhibit a tradeoff between the round complexity and the space per machine.
For $f$-approximation, they gave a $O((\phi / c)^2)$ round algorithm with space per machine $O(f n^{1 + c})$ by applying filtering~\cite{lattanzi2011filtering} to a primal dual algorithm. When the space per machine is nearly-linear, i.e., $c = O(1/\log n)$, this approach results in $O(\phi^2 \log^2 n)$ rounds, which is quadratically slower than our algorithm.

Bateni, Esfandiari, and Mirrokni~\cite{bateni2018optimal} developed a MapReduce algorithm for the $k$-cover problem that uses $\tO(n)$ space per machine. 
In this problem, one is given an integer $k$ and is asked to choose a family of at most $k$ sets that cover as many elements as possible. 
The problem, since it is a submodular maximization under cardinality constraint, admits a $\Theta(1)$-approximation. 
Their algorithm can be turned into an $O(\log n)$-approximate one for \setcover that uses $O(\log n)$ MPC rounds as follows.
%\kuba{Do you mean $O(\log^2 n)$ If so, this may be dominated by running MANIS (see my next comment)?}
Assume that $k$ is the minimum number of sets that covers all the elements; this assumption can be removed by making $O(\log n / \eps)$ guesses of the form $k = (1+\eps)^i$. 
Then, each time \cite{bateni2018optimal} is invoked, it covers a constant fraction of the elements. So, repeating that process $O(\log n)$ times covers all the elements using $O(k \cdot \log n)$ many sets. Our result provides tighter approximation and, when $\Delta \ll n$, also lower round complexity.

Since $k$-cover is a submodular maximization problem, the work \cite{liu2019submodular} yields $O(\log n)$ MPC round complexity and $O(\log n)$ approximation for \setcover. In the context of $k$-cover or \setcover, it is worth noting that the algorithm of \cite{liu2019submodular} sends $\Theta(\sqrt{n k})$ sets to a machine. 
It is unclear whether all those sets can be compressed to fit in $O(n)$ or smaller memory.

%\laxman{\cite{mirzasoleiman2016fast} and an earlier paper are mentioned in \cite{harvey2018greedy}, but the space and round-complexities are complicated and I am not sure it is worth stating here}

Ghaffari and Uitto~\cite{ghaffari2019sparsifying} developed a $\tO(\sqrt{\log \Delta})$ round complexity algorithm for \vcover in the sub-linear space regime. They first compute a maximal independent set, which is then used to obtain a maximal matching in the corresponding line graph. Finally, by outputting the endpoints of the edges in that maximal matchings, the authors provide a $2$-approximate \vcover. Our algorithm has a matching round complexity while, at the same time, it is arguably simpler.
For both $f$-approximation and $\hdelta$-approximation, we are unaware of any MPC algorithms that run in the sub-linear space regime. However, we note that the PRAM algorithm of Blelloch, Peng, and Tangwongsan~\cite{blelloch2011linear} can be simulated in this setting to obtain a round complexity of $O(\log^2 n)$ with $O(m)$ total space.

\paragraph*{$f$-Approximate \setcover{} in PRAM}
The first $f$-approximation algorithms for \setcover in the sequential setting are due to Hochbaum~\cite{hochbaum1982approximation}.
In the unweighted case, we can sequentially obtain an $f$-approximation in $O(m)$ work by picking any element, adding all of $\leq f$ sets containing it to the cover, and removing all newly covered elements.
For parallel algorithms aiming for $f$-approximation, Khuller, Vishkin, and Young~\cite{khuller1994primal} gave the first parallel $(1+\epsilon)f$-approximation for weighted \setcover that runs in $O(f m\log(1/\epsilon))$ work and $O(f \log^2 n \log(1/\epsilon))$ depth.
Their method uses a deterministic primal-dual approach that in each iteration raises the dual values $p(e)$ on every uncovered element $e$ until the primal solution, which is obtained by rounding every set $s$ where $\sum_{e \in s} p(e) \geq (1-\epsilon)w(s)$, is a valid set cover.
Their work analysis bounds the total number of times an element is processed across all $O(f\log n)$ iterations by $m$, giving a total work of $O(f\cdot m)$, which is not work-efficient. Their algorithm also has depth linear in $f$, which means that the number of iterations of their algorithm can be as large as $O(\log^2 n)$ for $f \leq \log n$, and the depth therefore as large as $O(\log^3 n)$.
%\kuba{I don't understand. Is f missing in this last bound?} \laxman{I plugged in the largest value that one would use the $f$ approximation instead of a $\hdelta$ approximation, which for $f = \log n$ gives $O(\log^2 n)$ depth. Maybe we should be more explicit about this. For larger $f$, the bound would be even worse (and eventually not even in $\mathsf{NC}$), but I don't think one would use an $f$ approximation in these cases, right?}

For weighted \vcover, Koufogiannakis and Young~\cite{koufogiannakis2011distributed} gave an elegant $2$-approximation that runs in $O(m)$ work in expectation and $O(\log n)$ depth.
They generalize their algorithm to work for $f$-approximate weighted \setcover in the {\em distributed} setting using Linial-Saks decomposition~\cite{panconesi1996complexity}; however, this does not imply an $\mathsf{NC}$ or $\mathsf{RNC}$ algorithm when $f > 2$.\footnote{$\mathsf{NC}$ contains all problems that admit log-space uniform circuits of polynomial size and poly-logarithmic depth and is the primary complexity class of interest when designing parallel algorithms. 
$\mathsf{RNC}$ extends $\mathsf{NC}$ by allowing the circuit access to randomness. By known simulation results~\cite{karp1988survey}, polynomial work and poly-logarithmic depth (randomized) PRAM algorithms also imply membership in $\mathsf{NC}$ ($\mathsf{RNC}$).}

Unlike the deterministic $(1+\epsilon)f$-approximation of Khuller et al.~\cite{khuller1994primal}, our algorithm is randomized and produces a set cover with the same approximation guarantees in expectation.
By contrast, our algorithm is easy to understand, analyze (with \cref{lem:main-sampling} as a given) and argue correctness.
Our algorithm is well suited for implementation and has small constant factors, since every element set or element and their incident edges are processed {\em exactly once} when the element is sampled or when the set is chosen.
We note that our algorithm also implies that $(1+\epsilon)f$-approximate \setcover is in $\mathsf{RNC}^1$ for any $f$; the work of~\cite{khuller1994primal} only implies this result for $f = O(1)$.
%\kuba{Let's define RNC}

%\cite{agarwal2018set}: Small neighborhood cover property for parallel and distributed \setcover.
%\cite{ran2022parallel}: Parallel min weight \setcover, following up on Agarwal. Another paper this group wrote this year claims the MANIS-based \setcover runs in $O(\log^2)$ PRAM rounds; need to check their RC claims.

\paragraph*{\matching and \hm in the Massively Parallel Computation Model}
The study of approximate matchings in MPC was initiated by Lattanzi et al.~\cite{lattanzi2011filtering}, who developed an $O(1)$ round algorithm for finding a maximal matching when the space per machine is $n^{1+\mu}$, for any constant $\mu > 0$. In the linear space regime, a line of work~\cite{czumaj2018round,ghaffari2018improved,assadi2019coresets,behnezhad2019exponentially} culminated in $O(\log \log n)$ MPC round complexity. In the sublinear space  regime, Ghaffari and Uitto~\cite{ghaffari2019sparsifying} developed a method that finds a maximal matching in $\tO(\sqrt{\log n})$ rounds.  When each machine has at least $O(n r)$ space, Hanguir and Stein~\cite{hanguir2021distributed} show how to find a maximal matching in $r$-hypergraph in $O(\log n)$ MPC rounds.
Their approach follows the filtering idea developed in \cite{lattanzi2011filtering}.
Our work does not only provide nearly quadratically lower round complexity compared to \cite{hanguir2021distributed}, but it also extends to the sub-linear space regime.

\paragraph*{$\hdelta$-approximate \setcover{} in PRAM}
Sequentially, $\hdelta$-approximate \setcover{} can be solved in $O(n + m)$ work by repeatedly selecting the set incident to the largest number of uncovered elements.
%Obtaining the sets chosen by this sequential algorithm is very challenging in parallel due to the sizes of sets depending on all previously chosen sets.
%Parallel algorithms for \setcover obtain parallelism by allowing some slack in which sets can be chosen, usually by bucketing the sets based on their size to within a factor of $(1+\epsilon)$, and processing these buckets one by one.
%%For simplicity, we focus on unweighted \setcover, so all sets in the same bucket have the same size up to a factor of $(1+\epsilon)$.
%Within a bucket, the sets can potentially have a large overlap, and simply picking all sets can result in $O(n)$ approximation. 
%The sets therefore need to be carefully chosen to ensure low overlap and different techniques have been developed to achieve this.
%
The first parallel approximation algorithm for \setcover was due to Berger, Rompel and Shor~\cite{berger1994efficient}, who gave a $(1+\epsilon)\hdelta$-approximation that runs in $O(m\log^5 n)$ work and $O(\log^5 n)$ depth whp.
Their algorithm buckets the sets based on their sizes into $O(\log \Delta)$ buckets.
It then runs $O(\log f)$ subphases, where the $j$-th subphase ensures that all elements have degrees at most $(1+\epsilon)^{j}$ (the subphases are run in decreasing order).
Each subphase performs $O(\log n)$ steps that work by either selecting sets that cover a constant fraction of certain large edges or otherwise independently sampling the remaining sets with probability $(1+\epsilon)^{-j}$.
Our approach also uses independent sampling but does not require handling two cases separately. As a result, our approach can be implemented efficiently by fixing the random choices upfront (see \cref{sec:fixing}).
%Overall, the Berger, Rompel and Shor algorithm guarantees a $(1+\epsilon)H_{\Delta}$ approximation, and requires $O(m\log^5 n)$ work and $O(\log^5 n)$ depth with high probability.
%
Subsequent work by Rajagopalan and Vazirani~\cite{rajagopalan1998primal} improved the work and depth, obtaining a parallel primal-dual algorithm with $O(m\log^3 n)$ work and $O(\log^3 n)$ depth with high probability, but a weaker approximation guarantee of $2(1+\epsilon)H_{\Delta}$. 

%They introduce a {\em permutation sampling} approach (see~\cite{tangwongsan2011efficient}) to handle selecting a (mostly) non-overlapping family of sets. 
%The idea is similar to Luby's MIS algorithm~\cite{luby1985simple}; it works by permuting the sets under consideration and having each element vote for the first set in the permutation that contains it. 
%Sets with a sufficient number of votes are then added into the \setcover.
%%Their algorithm is a parallel primal-dual algorithm for \setcover, and both the analysis and overall algorithm are simpler than those of~\cite{berger1994efficient}.

More recently, Blelloch, Peng and  Tangwongsan~\cite{blelloch2011linear} revisited parallel approximate \setcover with the goal of designing work-efficient algorithms.
Their algorithm achieves a $(1+\epsilon)\hdelta$-approximation in $O(m)$ expected work and $O(\log^3 n)$ depth with high probability on the CRCW PRAM.
%Their algorithm is therefore work-efficient and achieves the same depth as~\cite{rajagopalan1998primal}.
%Their algorithm formalizes and improves the permutation sampling ideas from~\cite{rajagopalan1998primal} to achieve work-efficiency.
%
They propose a general primitive inspired by the approach of~\cite{rajagopalan1998primal} called a Maximal Nearly-Independent Set (\manis{}), which, given a collection of sets chooses a subset of them while ensuring that the chosen sets are (1) nearly independent and thus do not have significant overlap, and (2) maximal, so that any unchosen sets have significant overlap with chosen ones.
%They give an algorithm for \manis{} with $O(m)$ expected work and $O(\log^2 n)$ depth whp which works by essentially applying permutation sampling $O(\log n)$ times.
%\kuba{I'd finish this section by saying how we improve MaNIS and stating our result again}
%
Blelloch, Simhadri, and Tangwongsan~\cite{blelloch2012parallel} later studied the algorithm in the Parallel Cache Oblivious model, and provided an efficient parallel implementation. 
%Subsequent work~\cite{dhulipala2017julienne, dhulipala2021theoretically} gave an improved implementation of~\cite{blelloch2011linear} and showed it can solve approximate \setcover on graphs with hundreds of billions of edges in minutes.

Compared to this prior work, we obtain a streamlined $(1+\epsilon)\hdelta$-approximate algorithm that shares some ideas with the previously discussed algorithms.
We also bucket the sets by size, and like \cite{rajagopalan1998primal, blelloch2011linear} each round finds a subset of sets with low overlap; the main difference is that our method is arguably simpler.
Our algorithm is also potentially very efficient in practice, since after we fix the randomness up-front (see \cref{sec:fixing}), we process every set in a bucket exactly once, unlike other implementations of \manis{} which can process a set within a bucket potentially many times~\cite{dhulipala2017julienne}.
%which require case analysis~\cite{berger1994efficient} in the symmetry breaking step to obtain ``nice'' instances of \setcover that are amenable to independent sampling, thus sacrificing logarithmic factors, or require more heavy-duty constructions such as \manis{}, our independent sampling approach is much simpler and is likely very practical.
Overall, our algorithm is work-efficient and runs in $O(\log^2 n \log\log n)$ depth on the CRCW PRAM.
Although this is an improvement over known depth bounds for PRAM algorithms, one can obtain similar bounds (in expectation) for the algorithm of~\cite{blelloch2011linear} by applying similar PRAM techniques.% (see \cref{sec:pram}).
We also note that both algorithms achieve $O(\log^3 n)$ depth in the binary-forking model~\cite{blelloch2020optimal}, and no parallel $\hdelta$-approximate algorithms exist with $o(\log^3 n)$ depth in this model.
Experimentally comparing our algorithm with existing implementations of~\cite{blelloch2011linear} is an interesting direction for future work.

%\cite{bercea2014computing}: MIS in hypergraphs in PRAM
%\cite{kuhn2018efficient}: MIS in linear hypergraphs

%\cite{koufogiannakis2011distributed}: Distributed Algorithms for Covering, Packing and Maximum Weighted Matching

%\cite{kuhn2006price}: ``The Price of Being Near-Sighted'' -- they get $O(\log n \cdot \log m)$ many iterations? They need to use unbounded messages and Linial-Saks.

%\cite{10.1145/3326171}: Deterministic parallel hypergraph MIS

\subsection{Outline}
The rest of the paper is organized as follows.
In \cref{sec:prelim} we introduce notation that we use in the paper.
\cref{sec:overview} contains a technical overview of our results. In particular, it describes our algorithms and outlines how they can be analyzed and efficiently implemented in the MPC and PRAM models.
Then, in \cref{sec:approximation}, we provide the approximation analysis of our basic algorithms.
\cref{sec:sampling-lemma} contains the analysis of a set sampling process, a key technical ingredient behind our algorithms.
%Due to the volume of existing work on the \setcover{} problem, we provide a broader overview of the related work in \cref{sec:related}.
%The approximation analysis of our basic algorithms can be found in \cref{sec:approximation}.
Finally, we describe and analyze our MPC algorithms in \cref{sec:MPC-algorithms} and our PRAM algorithms in \cref{sec:pram}.

\section{Preliminaries}\label{sec:prelim}
In the \setcover problem, we have a collection of \emph{elements} $T$ and a family of \emph{sets} $S$, which we can use to cover elements of $T$.
We represent an instance of the problem with a bipartite graph $G=((S \cup T), E)$, where $st \in E$ if and only if element $t$ belongs to the set $s$.
For a vertex $x \in S \cup T$ we use $N(x)$ to denote the set of its neighbors.
Since $G$ is bipartite, $x \in S$ implies $N(x) \subseteq T$ and $x \in T$ implies $N(x) \subseteq S$.
In particular, for $x \in S$, $|N(x)|$ is the size of the set $x$.\footnote{Technically, $|x|$ is also the size of the set $x$. However, in our algorithms, we repeatedly remove some elements from $T$ (together with their incident edges), and so we use $|N(x)|$ to make it clear that we refer to the \emph{current} size of the set.}
We use $\Delta$ to denote the maximum set size (i.e., the maximum degree of any vertex in $S$) and $f$ to denote the largest number of sets that contain any element (the maximum degree of any vertex in $T$).
Note that some of our algorithms modify the input graphs along the way, but we assume $\Delta$ and $f$ to be constant and refer to the corresponding quantities in the input graph.

The \vcover problem is defined as follows.
The input is an undirected graph $G=(V,E)$ and the goal is to find the smallest set $C \subseteq E$ such that each edge has at least one endpoint in $C$.
We note that this problem is equivalent to the \setcover problem in which each element belongs to exactly $2$ sets, except that the graph representing an instance is constructed a bit differently.

In the \hm problem, the input is a hypergraph $G$ consisting of a set of vertices $V$ and a set of edges $E$.
Each edge is a nonempty subset of $V$.
The \emph{rank} of a hypergraph $G$ is the maximum size of any edge.
In the \hm problem the goal is to find a subset $M \subseteq E$ which contains pairwise disjoint edges and has maximum possible size.
As opposed to \setcover, this is a maximization problem, and thus we say that the solution $M$ to the \hm problem is $\alpha$-approximate, for $\alpha \in (0, 1]$, when $|M| \geq \alpha \cdot |\textsc{OPT}|$, where $\textsc{OPT}$ is an optimal solution to the \hm problem.
The \matching problem is the \hm problem in simple graphs, i.e., in graphs with all edges of size $2$. 

\subparagraph*{Notation.}
We use $\tO(x)$ to hide logarithmic factor in $x$, i.e., $\tO(x)$ denotes $O(x \cdot \poly \log x)$. Throughout this paper, we use $n$ to refer to the number of vertices and $m$ to refer to the number of edges of an input graph. When it is stated that a guarantee holds ``with high probability'', or whp for short, it means that it holds with probability $1 - 1/n^c$, where $c$ is a constant. In our proofs with whp guarantees, $c$ can be made arbitrarily large by paying a constant factor in the round, space, total work, or depth complexity. Hence,  we often omit the exact value of $c$.

\subparagraph*{Probability tools.} In our analysis, we extensively apply the following well-known tool from probability.
\begin{theorem}[Chernoff bound]\label{lemma:chernoff}
	Let $X_1, \ldots, X_k$ be independent random variables taking values in $[0, 1]$. Let $X \eqdef \sum_{i = 1}^k X_i$ and $\mu \eqdef \E{X}$. Then,
	\begin{enumerate}[(A)]
%		\item\label{item:delta-at-most-1} For any $\delta \in [0, 1]$ it holds $\prob{|X - \mu| \ge \delta \mu} \le 2 \exp\rb{- \delta^2 \mu / 3}$.
		\item\label{item:delta-at-most-1-le} For any $\delta \in [0, 1]$ it holds $\prob{X \le (1 - \delta) \mu} \le \exp\rb{- \delta^2 \mu / 2}$.
		\item\label{item:delta-at-most-1-ge} For any $\delta \in [0, 1]$ it holds $\prob{X \ge (1 + \delta) \mu} \le \exp\rb{- \delta^2 \mu / 3}$.
%		\item\label{item:delta-at-least-1} For any $\delta \ge 1$ it holds $\prob{X \ge (1 + \delta) \mu} \le \exp\rb{- \delta \mu / 3}$.
	\end{enumerate}
\end{theorem}

\subparagraph*{Work-Depth Model.}
We study our algorithms in the shared-memory setting using the concurrent-read concurrent-write (CRCW) parallel random access machine model (PRAM).
We state our results in terms of their work and depth. 
The {\em work} of a PRAM algorithm is equal to the total number of operations required, and the {\em depth} is equal to the number of time steps required~\cite{jaja1992parallel}.
Algorithms with work $W$ and depth $D$ can be scheduled to run in $W/P + O(D)$ time~\cite{jaja1992parallel, blumofe1999scheduling} on a $P$ processor machine.
The main goal in parallel algorithm design is to obtain work-efficient algorithms with low (ideally poly-logarithmic) depth. A {\em work-efficient} algorithm asymptotically requires the same work as the fastest sequential algorithm.
Since the number of processors, $P$, is still relatively small on modern multicore machines, minimizing $W$ by designing work-efficient algorithms is critical in practice.
Our algorithms make use of several PRAM primitives, including parallel prefix sum~\cite{jaja1992parallel}, parallel integer sort~\cite{rajasekaran1989optimal}, and approximate prefix sums~\cite{hochbaum1982approximation}.

\section{Technical Overview}
\label{sec:overview}
In this section we demonstrate the main ideas behind our results.
We start by presenting our sequential algorithms for the \setcover problem.
For any set $X$ and probability $p \in [0, 1]$ we write $\sample{(X, p)}$ to denote a procedure that returns a random subsample of $X$ in which each element of $X$ is included independently with probability $p$.
Our algorithms work by repeatedly sampling sets or elements independently using a sequence of probabilities $p_i$, which is defined as follows for any $\epsilon > 0$.
\begin{align}\label{eq:define-p_i}
    b & \eqdef \lceil \log(2+2\epsilon) / \epsilon \rceil \\
    p_i & \eqdef (1+\epsilon)^{-\lceil i/b \rceil} \qquad \textrm{for any } i \in \mathbb{N}.\label{eq:p}
\end{align}
The intuition behind choosing $p_i$ and $b$ is discussed in \cref{sec:intuition-of-choosing-pi-and-b}.
Throughout the paper, we use $\log$ to denote the natural logarithm function.
We can now present our algorithms for \setcover, which are given as \cref{alg:f-apx} and \cref{alg:log-apx}.

\cref{alg:f-apx} is a natural parallelization of the sequential $f$-approximate algorithm.
Instead of picking one element at a time, we sample multiple elements at random and add the sets containing them to the solution.
The sampling probability is slowly increased in each step (or, more precisely, every $O(1/\epsilon)$ steps).
\cref{alg:log-apx} in turn parallelizes the $O(\log \Delta)$ approximate algorithm.
The outer loop iterates over different set sizes (rounded to the power of $1+\epsilon$) starting from the largest ones.
For a fixed set size, the inner loop adds to the solution a uniformly random sample of sets, again slowly increasing the sampling probability.

\begin{algorithm}[t]
\caption{$(1+\epsilon)f$-approximate algorithm for \setcover}\label{alg:f-apx}
\begin{algorithmic}[1]

\Function{SetCover}{$G, \epsilon$} \Comment{$G = (S \cup T, E)$}
    \State{$C \gets \emptyset$}
    \For{$i = b \lceil \log_{1+\epsilon} (\Delta / \epsilon)\rceil$ \textbf{down to} $0$}
        \State{$D \gets \sample(T, p_i)$}
        \State{$C \gets C \cup N(D)$}
        \State{Remove from $G$ all sets in $N(D)$ and all elements they cover}
    \EndFor
    \State \Return $C$
\EndFunction

\end{algorithmic}
\end{algorithm}

\begin{algorithm}[t]
\caption{$(1+\epsilon)\hdelta$-approximate algorithm for \setcover}\label{alg:log-apx}
\begin{algorithmic}[1]

\Function{SetCover}{$G, \epsilon$}\Comment{$G = (S \cup T, E)$}
    \State{$C \gets \emptyset$}
    \For{$j = \lfloor \log_{1+\epsilon}\Delta \rfloor$ \textbf{down to} $0$ \label{line:log-n-alg-loop-j}}
    \For{$i = b \lceil \log_{1+\epsilon} (f / \epsilon) \rceil$ \textbf{down to} $0$ \label{line:log-n-alg-loop-i}}
        \State{$D \gets \sample(\{s \in S \mid |N(s)| \geq (1+\epsilon)^j\}, p_i)$}
        \State{$C \gets C \cup D$}
        \State{Remove from $G$ all sets in $D$ and all elements they cover}
    \EndFor
    \EndFor
    \State \Return{C}
\EndFunction

\end{algorithmic}
\end{algorithm}

We start by analyzing \cref{alg:f-apx}.
Clearly, the algorithm runs in $O(\log n)$ iterations.
Iteration $i$ samples each element independently with probability  $p_i$ and adds all sets covering the sampled elements to the solution.
Then, all chosen sets and elements that became covered are removed.
Crucially, the sampling probability in the first step is at most $(1+\eps)^{-\lceil \log_{1+\epsilon} (\Delta / \epsilon)\rceil} \leq \epsilon / \Delta$, which implies that the expected number of elements sampled within each set is at most $\epsilon$.
Moreover, the sampling probability increases very slowly, as it increases by a $1+\epsilon$ factor every $b$ steps.

Let us now present the main ideas behind the analysis of the approximation ratio of the algorithm.
For simplicity of presentation, let us consider the case when each element is contained in exactly two sets (which implies $f=2$).
In other words, we consider the \vcover problem.
Specifically, since the degree of each vertex of $T$ is $2$, we can dissolve vertices of $T$ (equivalently, contract each such vertex into its arbitrary neighbor) and obtain a graph $H=(V,E)$ (where $V = S$) on which we would like to solve the \vcover problem.
We note that the solution and analysis of \cref{alg:f-apx} for \vcover generalizes easily to the case of arbitrary $f$.

If we translate \cref{alg:f-apx} to an algorithm running on $H$, we see that it repeatedly samples a set of \emph{edges} of $H$, and for each sampled edge $e$ adds both endpoints of $e$ to the solution, and removes both endpoints of $e$ from $H$ together with their incident edges.
In order to prove the approximation guarantee, we show the following.

\begin{lemma}\label{lem:aux-degree}
Let $D$ be the subset of $T$ picked across all iterations of \cref{alg:f-apx}.
For each vertex $v$, $\E{\deg_D(v))} \leq 1+O(\epsilon)$.
\end{lemma}

Here $\deg_D(v)$ denotes the number of elements of $D$ contained in $v$.
We prove this lemma formally in \cref{sec:approximation} (see \cref{lem:fapx-properties}).

Notice that when an element $x \in T$ is sampled to $D$ in \cref{alg:f-apx}, \emph{all} the sets containing $x$ are added to $C$.
So if $w$ sampled elements belong to the same set, then the algorithm could add $\Theta(w f)$ many sets, although only one of the sets suffices to cover all the $w$ sampled elements.
Intuitively, \cref{lem:aux-degree} states that the value of $w$ is at most $1 + O(\epsilon)$ in expectation, which we turn into an approximation guarantee in \cref{lem:f-apx}.

To prove \cref{lem:aux-degree}, we model the sampling process in the algorithm as follows.
Fix a vertex $v \in V$.
Let $A$ be the set of edges incident to $v$ in $G$.
\cref{alg:f-apx} runs a sequence of $b \lceil \log_{1+\epsilon} (\Delta / \epsilon)\rceil+1$ steps, indexed by $b \lceil \log_{1+\epsilon} (\Delta / \epsilon)\rceil, \ldots, 0$.
Note that the step indices are \emph{decreasing}.
Moreover, $\Delta \geq |A|$, since $\Delta$ is the maximum vertex degree in $G$.
In step $i$, each element of $A$ is sampled independently with probability $p_i$.
As soon as at least one element of $A$ is sampled, $v$ is added to the cover. When this happens all elements of $A$ are deleted and the random process stops.
Moreover, even if no element of $A$ is sampled, due to other random choices of the algorithm some elements of $A$ may get deleted. In particular, when a neighbor $w$ of $v$ is added to the set cover, the edge $wv$ is deleted from $A$.
Hence, it is possible that all elements of $A$ are deleted before any of them is sampled.

To analyze this we introduce the following \emph{set sampling process}, henceforth denoted as \ssp, which gives a more abstract version of the above sampling.  The analysis of this process forms our core sampling lemma, which will be useful not just for \cref{lem:aux-degree} but to analyze all of our algorithms. 
Let $A$ be a fixed set and $k$ be an integer.
The process proceeds in $k+1$ steps indexed $k, k-1, \ldots, 0$ and constructs a family of sets $A = A_k \supseteq A_{k-1} \supseteq \dots \supseteq A_1 \supseteq A_0$ as well as a family $R_k, \ldots, R_0$, such that $R_i \subseteq A_i$.
In each step $i$ ($k \geq i \geq 0$) we first construct a set $A_i$.
We have that $A_k = A$, and for $i \in [0, k)$ the set $A_i \subseteq A_{i+1}$ is constructed by a (possibly randomized) \emph{adversary}, who is aware of the sets $A_j$ and $R_j$ for $j > i$.
Our analysis of \ssp aims to argue that certain guarantees hold regardless of what the adversary does.
After the adversary constructs $A_i$, we sample $R_i = \sample(A_i, p_i)$.

We note that we assume that the updates are adversarial to simplify the overall proof.
This makes our claims about \ssp more robust, and analyzing \ssp with an adversary does not introduce significant complications. 

For $i \in [0, k]$, we define $n_i \eqdef |A_i|$.
Whenever we apply \ssp we have that $k \geq b \lceil \log_{1+\epsilon}(n_k / \epsilon)\rceil$, and for simplicity we make this assumption part of the construction.
Note that this condition simply ensures that the initial sampling probability is at most $\epsilon / n_k$.
Finally, we let $z$ be the maximum index such that $R_z \neq \emptyset$.
We stress that the \ssp steps are indexed in {\em decreasing order}, and hence $z$ is the index of the \emph{first} step such that $R_z$ is nonempty.
If all $R_i$ are empty, we set $z = -1$ and assume $R_{-1} = \emptyset$.
We say that $z$ is the step when the \ssp stops.

Observe that to analyze the properties of the set of sampled edges in \cref{alg:f-apx}, it suffices to analyze the properties of the set $R_z$.
Our main lemma analyzing \ssp is given below.
It captures the single property of \ssp which suffices to prove the approximation ratio of both \cref{alg:f-apx} and \cref{alg:log-apx}.
In particular, it directly implies \cref{lem:aux-degree}.

\begin{restatable}{lemma}{mainsampling}\label{lem:main-sampling}
Consider the \ssp using any adversary and $\epsilon > 0$. Then, $\E{|R_{z}|} \leq 1+4\epsilon$.
\end{restatable}

Let us now describe the intuition behind the proof of \cref{lem:main-sampling}.
To simplify the presentation, let us assume that the sets $A_0, \ldots, A_k$ are fixed upfront (i.e., before any set $R_i$ is sampled).
We show in \cref{obs:nonadaptive} that if we are interested in analyzing the properties of $R_z$, this can be assumed without loss of generality.  Observe that as long as the process executes steps where $p_i \cdot n_i \leq \epsilon$, the desired property holds.
Indeed, with this assumption we have that $\E{|R_i| \mid R_i \neq \emptyset} \leq 1+\epsilon$.
This is because even if one element is sampled, the expected size of the sample among all remaining elements is at most $\epsilon$ (for a formal proof, see \cref{cl:condition-nonempty}).

In order to complete the proof, we show that reaching a step where $p_i \cdot n_i \gg \epsilon$ is unlikely.
Specifically, the value of $p_i \cdot n_i$ can increase very slowly in consecutive steps, as $p_i$ increases only by $(1+\epsilon)$ factor every $b$ steps, and $n_i$ can only decrease.
By picking a large enough value of $b$, we can ensure that the process most likely stops before $p_j \cdot n_j$ becomes large, i.e., the expected value of $p_z \cdot n_z$ is $O(\epsilon)$.
Indeed, in each step where $p_i \cdot n_i \geq \epsilon$, the process stops with probability $\Omega(\epsilon)$.
Hence, if we repeat such a step roughly $1/\epsilon$ times (which can be achieved by tweaking $b$), the process will stop with constant probability (independent of $\epsilon$).
In the end we fix $b$, such that the probability of $p_i \cdot n_i$ increasing by a factor of $1+\epsilon$ is at most $1/(2+2\epsilon)$.
As a result, thanks to a geometric sum argument, the expected value of $p_z \cdot n_z$ is $O(\epsilon)$, which implies \cref{lem:main-sampling}.

\subsection{Fixing the Random Choices Upfront}\label{sec:fixing}

In order to obtain efficient implementations of our algorithms, we reformulate them into equivalent versions where the sampling happens upfront.
Specifically, consider the main loop of \cref{alg:f-apx}.
Observe that each element is sampled at most once across all iterations, since as soon as an element is sampled it is removed from further consideration.
A similar property holds for each set across all iterations of the inner for loop of \cref{alg:log-apx}.
Moreover, in both cases, the probability of being sampled in a given iteration is fixed upfront and independent of the algorithm's actions in prior iterations.
It follows easily that we can make these per-element or per-set random choices upfront.
Specifically, let $k = b \lceil \log_{1+\epsilon} (\Delta / \epsilon)\rceil$. 
Then, \cref{alg:f-apx} executes $k+1$ iterations indexed  $k, k-1, \ldots, 0$.
We can randomly partition the input elements into $k+1$ \emph{buckets} $B_k, \ldots, B_0$ using a properly chosen distribution and then in iteration $i$ consider the elements of $B_i$ which have not been previously removed as the sample to be used in this iteration.

Observe that since $p_0 = 1$ (see \cref{eq:p}), each element that is not removed before the last step is sampled.
Specifically, let $\tilde{p_0}, \ldots, \tilde{p_k}$ be a probability distribution such that $\tilde{p_i} = p_i \cdot \prod_{j=i+1}^k (1-p_j)$.
Observe that $\tilde{p_i}$ is the probability that an element should be put into bucket $B_i$.

\cref{alg:f-apx-fixed} shows a version of \cref{alg:f-apx} in which the random choices are made upfront.
It should be clear that \cref{alg:f-apx,alg:f-apx-fixed} produce the same output.
Moreover, an analogous transformation can be applied to the inner loop of \cref{alg:log-apx}.
The benefit of making the random choices upfront is twofold.
In the MPC model, we use the sampling to simulate $r$ iterations of the algorithms in $O(\log r)$ MPC rounds.
The efficiency of this simulation crucially relies on the fact that we only need to consider the edges sampled within the phase and we can determine (a superset of) these edges upfront.

In the PRAM model, the upfront sampling allow us to obtain an improved work bound: instead of tossing a coin for each element separately in each iteration, we can bucket the elements initially and then consider each element in exactly one iteration.
In order to bucket the elements efficiently we can use %the so-called Alias method~\cite{walker1974new}, summarized by 
the following lemma.

\begin{lemma}[\cite{walker1974new}]
\label{lem:aliasmethod}
Let $r_0, r_1, \ldots, r_k$ be a sequence of nonnegative real numbers which sum up to $1$.
Let $X \rightarrow [0, k]$ be a discrete random variable, such that for each $i \in [k]$, $P(X = i) = r_i$.
Then, there exists an algorithm which, after preprocessing in $O(k)$ time, can generate independent samples of $X$ in $O(1)$ time.
\end{lemma}

\begin{algorithm}[t]
\caption{$(1+\epsilon)f$-approximate algorithm for \setcover}\label{alg:f-apx-fixed}
\begin{algorithmic}[1]

\Function{SetCover}{$G, \epsilon$}\Comment{$G = (S \cup T, E)$}
    \State{$C \gets \emptyset$}
    \State{$k \gets b \lceil \log_{1+\epsilon}
    (\Delta / \epsilon)\rceil$}
    \State{$B_i \gets \emptyset$ for all $i \in [0, k]$}
    \For{\textbf{each} element $t \in T$}
        \State{Sample $X_t \in [0, k]$, where $P(X_t = i) = \tilde{p_i}$ and add $t$ to $B_{X_t}$}
    \EndFor
    
    \For{$i = k$ \textbf{down to} $0$}
        \State{$D \gets $ all elements of $B_i$ which are not marked}
        \State{$C \gets C \cup N(D)$}
        \State{Remove from $G$ all sets in $N(D)$ and mark all elements they cover}
    \EndFor
    \State \Return $C$
\EndFunction

\end{algorithmic}
\end{algorithm}

\subsection{MPC Algorithms}

Simulating \cref{alg:f-apx-fixed} in the sub-linear MPC model is a relatively straightforward application of the graph exponentiation technique~\cite{lenzen2010brief,ghaffari2019sparsifying,ghaffari2019improved,chang2019complexity,brandt2021breaking}.
For simplicity, let us again consider the \vcover problem.
We will show how to simulate $r = O(\sqrt{\log n})$ iterations of the for loop in only $O(\log r) = O(\log \log n)$ MPC rounds.
Let us call these $r$ iterations of the algorithm a \emph{phase}.
We first observe that to execute a phase we only need to know edges in the buckets corresponding to the iterations within the phase.
Let us denote by $G_r$ the graph consisting of all such edges.
Moreover, let $p$ be the sampling probability used in the first iteration of the phase.
The crucial observation is that the maximum degree in $G_r$ is $2^{\tilde{O}(\sqrt{\log n})}$ with high probability.
This can be proven in three steps.
First, we show that by the start of the phase the maximum degree in the original graph $G$ drops to $O(1/p \cdot \log n)$ with high probability.
Indeed, for any vertex $v$ with a higher degree the algorithm samples an edge incident to $v$ with high probability, which causes $v$ to be removed.
Second, we observe that the sampling probability increases to at most $p \cdot 2^{O(\sqrt{\log n})}$ within the phase, and so the expected number of edges incident to any vertex of $G_r$ is at most $O(1/p \cdot \log n) \cdot p \cdot 2^{O(\sqrt{\log n})} = 2^{\tilde{O}(\sqrt{\log n})}$.
Third, we apply a Chernoff bound.

At this point, it suffices to observe that running $r$ iterations of the algorithm can be achieved by computing for each vertex $v$ of $G_r$ a subgraph $S_v$ consisting of all vertices at a distance $O(r)$ from $v$ and then running the algorithm separately on each $S_v$.
In other words, running $r$ iterations of the algorithm is a $O(r)$ round \textsc{Local} algorithm.
Computing $S_r$ can be done using graph exponentiation in $\log r = O(\log \log n)$ MPC rounds using $2^{O(\sqrt{\log n}) \cdot r} = n^{\alpha}$ space per machine and $n^{1+\alpha}$ total space, where $\alpha > 0$ is an arbitrary constant.

The space requirement can also be reduced to $\tilde{O}(m)$.
We now sketch the high-level ideas behind this improvement.
We leverage the fact that if we sample each edge of an $m$-edge graph independently with probability $p$, then only $O(p \cdot m)$ vertices have an incident sampled edge, and we can ignore all the remaining vertices when running our algorithm.
Hence, we only need to run the algorithm for $O(p \cdot m)$ vertices and thus have at least $S = m / O(p \cdot m) = \Omega(1/p)$ available space per vertex, even if we assume that the total space is $O(m)$.
As argued above, with space per vertex $S$, we can simulate roughly $\sqrt{\log S}$ steps of the algorithm.
In each of these steps, the sampling probability increases by a constant factor, so overall, it increases by a factor of $2^{\Omega(\sqrt{\log S})}$ across the $\sqrt{\log S}$ steps that we simulate.
After repeating this simulation $t = \sqrt{\log S} = O(\sqrt{\log \Delta})$ times, the sampling probability increases by a factor of at least $2^{\Omega(t \cdot \sqrt{\log S})} = 2^{\Omega(\log S)} = S^{\Omega(1)}$.
Overall, after roughly $O(\sqrt{\log \Delta})$ repetitions the space per vertex reduces from $S$ to $S^{1-\Omega(1)}$.
Similarly, the sampling probability increases from $p$ to $p^{1+\Omega(1)}$.
Hence, it suffices to repeat this overall process $\log \log \Delta$ times to simulate all $O(\log \Delta)$ steps.
We delve into details of this analysis in \cref{sec:MPC-algorithms}.

\subsection{PRAM algorithms}

\cref{alg:f-apx-fixed} also almost immediately yields a work-efficient algorithm with $O(\log n)$ depth in the CRCW PRAM.
Obtaining a work-efficient and low-depth implementation of \cref{alg:log-apx} is only a little more involved.
One challenge is that the set sizes change as elements get covered.
Since we run $O(\log n)$ steps per round, we can afford to exactly compute the sizes at the start of a round, but cannot afford to do so on every step without incurring an additional $O(\log n)$ factor in the depth.
We first use the randomness fixing idea described in \cref{sec:fixing} to identify the step in the algorithm when a set will be sampled.
Then, in every step, for the sets sampled in this step, we approximate the set sizes up to a $(1+\delta)$ factor, which can be done deterministically and work-efficiently in $O(\log\log n)$ depth and use these estimates in our implementation of \cref{alg:log-apx}.
The resulting algorithm still obtains a $(1+\epsilon)\hdelta$-approximation in expectation while deterministically ensuring work efficiency and $O(\log^2 n \log\log n)$ depth.

\begin{algorithm}[t]
\caption{Algorithm for \hm}\label{alg:approx-matching}
\begin{algorithmic}[1]

\Function{\hm}{$G, \epsilon$}\Comment{$G=(V,E)$}
    \State $\Delta \gets $ the maximum degree in $G$
    \State $C \gets \emptyset$
    \For{$i = b \lceil \log_{1+\epsilon} (\Delta / \epsilon)\rceil$ \textbf{down to} $0$}
        \State{$D \gets \sample(E(G), p_i )$}
        \State{$C \gets C \cup D$}
        \State{Remove from $G$ all endpoints of edges in $D$}
        
    \EndFor
    \State \Return edges independent in $C$

\EndFunction

\end{algorithmic}
\end{algorithm}

\subsection{\hm in MPC}
We show that our techniques for solving \setcover can be further applied to solve approximate \hm.
For the purpose of this high-level overview, we consider the special case of approximate \matching in simple graphs, i.e., hypergraphs in which each edge has exactly 2 endpoints.
Generalizing our approach to arbitrary hypergraphs does not require any additional ideas.
Our algorithm for \hm is shown as \cref{alg:approx-matching}, and
works similarly to \cref{alg:f-apx}.
Specifically, if we consider the simple graph setting and the \vcover problem, \cref{alg:f-apx} samples a set of edges of the graph and then returns the set of endpoints of these edges as the solution.
\cref{alg:approx-matching} also samples a set of edges, but the difference is in how it computes the final solution.
Namely, it returns all sampled edges which are independent, i.e., not adjacent to any other sampled edge.
Clearly, the set of edges returned this way forms a valid matching.
To argue about its cardinality, we show that the number of edges returned is a constant factor of all edges sampled.
To this end, we show a second fact about the \ssp, which says that any sampled element is \emph{not} sampled by itself with only a small constant probability.

\begin{restatable}{lemma}{singlesample}\label{lem:single-sample}
Let $a \in A_k$ and let $\epsilon \leq 1/2$.
Then $P(|R_z| > 1 \mid a \in R_z) \leq 6\epsilon$.
\end{restatable}

The high-level idea behind the proof of \cref{lem:single-sample} is similar to the proof of \cref{lem:aux-degree}: in the steps where the expected number of sampled elements is $\leq \epsilon$, the property follows in a relatively straightforward way.
Moreover, we are unlikely to reach any step where the expected number of sampled elements is considerably larger, and so to complete the proof we also apply a geometric sum-based argument.
With the above Lemma, the analysis of \cref{alg:approx-matching} becomes straightforward and shows that the approximation ratio of the algorithm is $\frac{1 - h \cdot 6\epsilon}{h}$ (see \cref{lemma:approx-matching}), where $h$ is the rank of the hypergraph.

% We note that the difference between \cref{lem:single-sample} and \cref{lem:main-sampling} is that \cref{lem:single-sample} shows that a certain property holds for each element of $A$.

\cref{alg:approx-matching} can be seen as a simplification of the ``warm-up'' algorithm of \cite{ghaffari2019sparsifying}, which alternates between sampling edges incident to high-degree vertices and peeling high-degree vertices. 
Our algorithm simply samples from \emph{all edges} and does not peel vertices. This makes the proof of the approximation ratio trickier since there is less structure to leverage. 
However, the simplification results in a straightforward application of round compression and enables extending the algorithm to hypergraphs.

\subsection{An intuition behind choosing $p_i$ and $b$}
\label{sec:intuition-of-choosing-pi-and-b}
To see why we choose $p_i$ and $b$ as described by \cref{eq:p}, imagine the following process. 
Initially, there are $n$ items, and we process them in rounds. 
In round $i$, we sample each item independently with probability $p_i$. 
After each round, some items disappear. Once we successfully sample at least one item, the process stops. We aim to choose a not-too-long and non-decreasing sequence of probabilities $p_i$, so that the expected number of items chosen when the process stops is at most $1 + \epsilon$.

Intuitively, we get into a bad situation in this process when the number of remaining items times $p_i$ is more than $\epsilon$. In particular, the desired bound holds if we never get into a bad situation.

A natural starting point is to start with a very low $p_i$, e.g., $\epsilon / n$, and increase $p_i$ by a factor of $(1+\epsilon)$ every step. 
However, it turns out that this growth is too fast, and it is quite likely that we will get into a bad situation. 
Hence, we slow the growth of $p_i$ and increase it by a $(1+\epsilon)$ factor every $b$ steps. Let us now discuss how to choose $b$.

Let us assume that we are getting \emph{close} to a bad situation, that is, $p_i$ times the number of remaining items is a bit more than, say, $\epsilon / 10$. Then, if we set $b = 1/\epsilon$, one of three things may happen in one round:
\begin{itemize}
    \item Many items have been removed, so we are no longer in a bad situation.
    \item Otherwise, not many items are removed, and so we repeat the sampling step $b = 1/\epsilon$ times. In each of these steps, the expected number of sampled items is $\Theta(\epsilon)$. Thus, with constant probability (independent of $\epsilon$), we successfully sample an item, and the process finishes; since we have not gotten into a bad situation yet, this is a good outcome.
    \item Otherwise, we do not sample any items in $b$ steps. By tweaking $b$, we can ensure this happens with, at most, say, $1/100$ probability. After these $b$ steps, $p_i$ increases by a factor of $(1+\epsilon)$, so we get closer to a bad situation. However, since this only happens with a probability of $1/100$, the probability of actually reaching a bad situation can be upper-bounded.
\end{itemize}

\section{Approximation Analysis of the Algorithms}\label{sec:approximation}

\begin{figure}[t]
\begin{minipage}[b]{0.45\linewidth}
\centering
\begin{align*}
\textrm{minimize}\qquad & \sum_{s \in S} x_s & \\
\textrm{subject to} \qquad & \sum_{s \ni t} x_s \geq 1 & \textrm{for } t \in T\\
& x_s \geq 0 & \textrm{for } s \in S
\end{align*}
%\vspace{-2em}
\end{minipage}
\hfill
%\hspace{3in}
\begin{minipage}[b]{0.45\linewidth}
\centering
\begin{align*}
\textrm{maximize} \qquad & \sum_{t \in T} y_t & \\
\textrm{subject to} \qquad & \sum_{t \in s} y_t \leq 1 & \textrm{for } s \in S\\
& y_t \geq 0 & \textrm{for } t \in T
\end{align*}
\end{minipage}
\caption{LP relaxation of the \setcover LP (left) and its dual (right).}\label{lp}
\end{figure}

In this section, we analyze the approximation ratio and analysis of our algorithms.
We note that the correctness of all algorithms is essentially immediate.
Specifically, in \cref{alg:f-apx} and \cref{alg:log-apx} we only remove an element when it is covered, and in the last iteration of the inner \textbf{for} loop in both algorithms (which in \cref{alg:f-apx} is the only loop) the sampling probability is $1$, so we add all remaining sets (in \cref{alg:log-apx}, limited to the large enough size) to the solution.
Similarly, \cref{alg:approx-matching} clearly outputs a valid matching, thanks to the final filtering step in the \textbf{return} statement.

\begin{lemma}\label{lem:fapx-properties}
Let $\bar{D}$ be the union of all elements picked in all iterations of \cref{alg:f-apx}.
For each set $s \in S$ we have $\E{|s \cap \bar{D}|} \leq 1+4\epsilon$.
\end{lemma}
\begin{proof}
This follows from \cref{lem:main-sampling} applied to the set $s$.
\end{proof}

\begin{lemma}\label{lem:f-apx}
\cref{alg:f-apx}, called with $e' = \epsilon / 4$, computes an $(1+\epsilon)f$-approximate solution to \setcover.
\end{lemma}
\begin{proof}
The key property that we utilize in the analysis is stated in \cref{lem:fapx-properties}.
The proof is a relatively simple generalization of the dual fitting analysis of the standard $f$-approximate \setcover algorithm.
The generalization needs to capture two aspects: the fact that the property stated in \cref{lem:fapx-properties} holds only in expectation and allows for a slack of $4\epsilon$.

We use the relaxation of the \setcover IP and its dual given in \cref{lp}.
Let $\bar{D}$ be the union of all elements picked in all iterations of \cref{alg:f-apx}.
We construct a dual solution that corresponds to $\bD$ as follows.
First, for each $t \in \bar{D}$, we set $\bar{y_t} = 1 / (1 + 4\epsilon)$, and for $t \not\in \bar{D}$ we set $\bar{y_t} = 0$.
Recall that \cref{alg:f-apx} returns a solution of size $|C|$.
For any run of the algorithm we have $|C| \leq f \sum_{t\in T} \bar{y_t} (1+4\epsilon)$.%, which implies

We now define a set dual of variables by setting $y_t = \E{\bar{y_t}}$ for each $t \in T$.
This set forms a feasible dual solution, since for every $s \in S$ we have
\[
\sum_{t \in s} y_t = \sum_{t \in s} \E{\bar{y_t}} = \E{|s \cap \bar{D}| /(1+4\epsilon)} \leq 1,
\]
where in the last inequality we used \cref{lem:fapx-properties}.
Moreover, we have
\[
\E{|C|} \leq f \cdot \sum_{t \in T} y_t (1+4\epsilon),
\] 
which implies that the solution's expected size is at most $f(1+4\epsilon)$ times larger than a feasible dual solution.
Hence, the lemma follows from weak LP duality.
\end{proof}

Now let us consider \cref{alg:log-apx}.

\begin{lemma}\label{lem:logn-properties}
\cref{alg:log-apx} adds sets to the solution in batches.
When a batch of sets $D$ is added to the solution we have that (a) the residual size of each set in $D$ is at most $(1+\epsilon)$ smaller than the maximum residual size of any set at that moment, and (b) for each newly covered element $t$, the expected number of sets in a batch that cover it is at most $(1+4\epsilon)$.
\end{lemma}

\begin{proof}
Observe that for $i = 0$ and the current value of $j$, each of the remaining sets of size $(1 + \eps)^j$ or more is included in $D$.
By applying this observation inductively, we see that each iteration of the outer loops starts with the maximum set size being less than $(1+\epsilon)^{j+1}$ and results in all sets of size at least $(1+\epsilon)^j$ being either added to the solution or removed from the graph.
This implies claim (a).
Claim (b) follows directly from \cref{lem:main-sampling}.
\end{proof}

To bound the approximation guarantee of \cref{alg:log-apx}, we show the following, which, similarly to the proof of \cref{lem:f-apx}, uses a dual-fitting analysis.

\begin{lemma}\label{lemma:set-cover-log-approx}
Any algorithm that computes a valid \setcover solution and satisfies the property of \cref{lem:logn-properties}, computes an $(1+\epsilon)(1+4\epsilon) \hdelta$-approximate (in expectation) solution to \setcover.
\end{lemma}

\begin{proof}
Similarly to the proof of \cref{lem:f-apx}, we use the standard dual-fitting analysis of \setcover and based on the properties of the algorithm captured in \cref{lem:logn-properties}.
The standard analysis does not work out of the box due to additional slack in both claims stated in the lemma (i.e. the fact that $\epsilon > 0$) and the expected guarantee in claim (b).

We use the LP relaxation of the \setcover IP and its dual shown in \cref{lp}.
Our goal is to define a feasible set of dual variables.
We first define a \emph{price} $p_t$ of each element $t$ as follows.
Initially, the price of each element is $0$.
When a set $s$ with $m$ uncovered elements $t_1, \ldots, t_m$ is added to the solution, we add $1/m$ to the prices of $t_1, \ldots, t_m$.
Note that we can increase each price multiple times this way, as the sets added to the solution in our algorithm may (typically only slightly) overlap.
Clearly, the sum of the prices defined this way is precisely the size of the \setcover solution.

Fix a set $s$.
Let $t_1, \ldots, t_{|s|}$ be the elements of $s$, ordered in the same way as they got covered in the algorithm and $p_1, \ldots, p_{|s|}$ be their corresponding prices.
In particular, just before $t_1$ was covered $s$ had $|s|$ uncovered elements.
Moreover, let $X_i$, for $i \in [1, |s|]$ denote the number of sets covering $t_i$ that have been added to the solution when $t_i$ was first covered.
%By \cref{lem:logn-properties} we have $\E{X_i} \leq 1+4\epsilon$.

\begin{claim}\label{cl:prices}
It holds that
$\sum_{i=1}^{|s|} \E{p_i} \leq (1+4\epsilon)(1+\epsilon) H_{|s|}$.

\end{claim}
\begin{proof}
By the ordering of elements, when element $t_i$ was first covered there were still $|s|-i+1$ uncovered elements of $|s|$, and so the algorithm was free to choose and add to the solution a set of at least this size.
Thus, by \cref{lem:logn-properties} the algorithm covered $t_i$ by a set of size at least $(|s|-i+1) / (1+\epsilon)$.
Hence, by the definition of the price, we have $p_i \leq \frac{X_i \cdot (1+\epsilon)}{|s|-i+1}$. Moreover, by \cref{lem:logn-properties}, $\E{X_i} \leq 1+4\epsilon$.
By using these two bounds, we obtain
\[
\sum_{i=1}^{|s|} \E{p_i} \leq \sum_{i=1}^{|s|} \frac{\E{X_i} (1+\epsilon)}{|s| - i + 1} \leq \sum_{i=1}^{|s|} \frac{(1+4\epsilon)(1+\epsilon)}{|s| - i + 1} \leq (1+4\epsilon)(1+\epsilon) H_{|s|}.
\]
\end{proof}

We now define the dual variable $y_t$ of each element $t$ as 
\[
y_t = \frac{\E{p_t}}{(1+4\epsilon)(1+\epsilon) H_{\Delta}}.
\]
By \cref{cl:prices}, the sum of the dual variables of elements in each set is at most $1$ (since $\Delta$ is the size of the largest set), and each dual variable is nonnegative.
As a result, these dual variables give a feasible solution to the dual of the \setcover problem.
Moreover, the expected number of sets in the solution is exactly $(1+4\epsilon)(1+\epsilon) H_{\Delta}$ factor larger than the sum of the dual variables (which is the objective of the dual LP).
Hence, the lemma follows from weak duality.
\end{proof}

\begin{lemma}
\label{lemma:approx-matching}
\cref{alg:approx-matching} ran on a rank $h$ hypergraph in expectation computes a $\frac{1 - h \cdot 6\epsilon}{h}$ approximate matching.
\end{lemma}

\begin{proof}
Let $G = (V, E)$ be input to \cref{alg:approx-matching}, and let $C'$ be the set $C$ after executing the for-loop. Observe that $V(C')$ is a vertex cover of $G$: all the edges not covered by the time we reach $i = 0$ are included in $D$ and, so, in $C$.

We also want to lower-bound the size of independent edges in $C'$. 
Fix an edge $e$ and consider a vertex $v \in e$. Once $e$ is included in $C$, all the endpoints of $e$ are removed from $G$. Hence, if $e$ is not independent in $C'$, then it is the case because, in the same iteration, an edge $e'$ adjacent to $e$ is also included in $D$.
To upper-bound the probability of $e'$ and $e$ being included in $D$, we use \cref{lem:single-sample}. How do we use \cref{lem:single-sample} in the context of \cref{alg:approx-matching}?
For a fixed vertex $v$, $A_i$ is the set of edges incident to $v$ at the $i$-th iteration of the for-loop of \cref{alg:approx-matching}. In particular, the set $A_k = A$ defined in \cref{sec:sampling-lemma} equals all the edges of the input graph $G$ containing $v$.

By \cref{lem:single-sample}, the probability of $v$ being incident to more than one sampled edge is at most $6\epsilon$.
Thus, by union bound, $e$ and an edge adjacent to $e$ are included in $D$ with probability at most $h \cdot 6\epsilon$. Therefore, with probability $1 - h \cdot 6 \eps$ at least, a fixed edge in $C'$ is independent. This implies that in expectation $|C'| (1 - h \cdot 6 \eps)$ edges in $C'$ are independent and, so, \cref{alg:approx-matching} outputs a matching that in expectation has size at least $|C'| (1 - h \cdot 6 \eps)$. Since there is a vertex cover of size $|V(C')| \le h |C'|$ at most, it implies that \cref{alg:approx-matching} in expectation produces a $\frac{1 - h \cdot 6\epsilon}{h}$-approximate maximum matching.
%The set of vertices included in the initially chosen edges forms a vertex cover, and so its size is at least $1/h$ the size of a maximum matching.
\end{proof}

\section{Analysis of the Set Sampling Process}
\label{sec:sampling-lemma}
\iffalse
In this section, our goal is to analyze the following random process.
Let $A$ be a nonempty set and $0 < \epsilon < 1/2$.
The process consists of several steps.
In a step, we sample each element of $A$ independently with the same probability.
After that, some elements of $A$ may randomly disappear.
Before the next step begins, the sampling probability is (slightly) increased.
The process finishes when either $A$ becomes empty, or we sample at least one element.
Our goal is to show that once we assume that the initial sampling probability is sufficiently low and the probabilities increase slowly enough, then the two following properties hold.
First, if the process ends due to an element being sampled, then the expected number of sampled elements is $1 + O(\epsilon)$ (see \cref{lem:main-sampling}).
Second, each sampled element is \emph{not} the only element of the sample with probability $1-O(\epsilon)$ (see \cref{lem:single-sample}).
\fi

In this section, we prove two properties of the \ssp, which are key in analyzing our algorithms. Recall the set sampling process SSP: initially, $A_k = A$, and $R_k$ is obtained by including each element of $A_k$ independently with probability $p_k$.  Then, for every $i$ from $k-1$ down to $0$, an adversary chooses $A_i \subseteq A_{i+1}$ (possibly with randomization) and then $R_i$ is obtained by including every element of $A_i$ independently with probability $p_i$.

Note two things about this process.  First, the adversary can be randomized.  Second, the adversary can be \emph{adaptive}: its choice of $A_i$ can depend on $A_j$ and $R_j$ for $j > i$.  
Recall that after running this process, $z$ is the largest index such that $R_z$ is nonempty.  Our goal in this section is to prove the following lemma:

\mainsampling*

To get some intuition for why~\cref{lem:main-sampling} might be true, observe that the sampling probability $p_i$ increases very slowly; specifically, it increases by a $(1+\epsilon)$ factor every $b$ steps. So the algorithm gets many chances at each (low) probability to obtain a non-empty $R_i$, and so it is not very likely to get more than $1$ element in $R_z$.

To prove this lemma, we start with a simple but extremely useful observation: we may assume without loss of generality that the adversary is \emph{nonadaptive}: its choice of $A_i$ does not depend on $R_j$ for $j > i$ (it can still depend on $A_j$ for $j > i$).
In other words, a nonadaptive adversary must pick the entire sequence of $A_i$'s before seeing the results of any of the $R_i$'s.
Moreover, we may assume that the adversary is \emph{deterministic}.

\begin{observation} \label{obs:nonadaptive}
    Without loss of generality, the adversary is nonadaptive and deterministic, i.e., it is a single fixed sequence $A_k, A_{k-1}, \dots, A_0$.
\end{observation}
\begin{proof}
    We begin by showing that the adversary is nonadaptive without loss of generality.  To see this, suppose there is some adaptive adversary $P$.  Then let $P'$ be the nonadaptive adversary obtained by simply running $P$ under the assumption that every $R_i = \emptyset$.  Clearly, this gives a (possibly randomized) sequence $A_k, A_{k-1}, \dots, A_0$ without needing to see the $R_i$'s, and so is nonadaptive. Clearly, $P$ and $P'$ behave identically until $z-1$, i.e., until just \emph{after} the first time that some $R_i$ is nonempty (since $P'$ is just $P$ under the assumption that all $R_i$'s are empty).  But indices $z-1$ down to $0$ make no difference in \cref{lem:main-sampling}!  Hence if \cref{lem:main-sampling} holds for nonadaptive adversaries, it also holds for adaptive adversaries.

    So we assume that the adversary is nonadaptive, i.e., the adversarial choice is simply a distribution over sequences $A_k, A_{k-1}, \dots, A_0$.  This means that the expectation in \cref{lem:main-sampling} is taken over both the adversary's random choices and the randomness from sampling the $R_i$'s once the $A_i$'s are fixed.
    These are intermixed for an adaptive adversary but for a nonadaptive adversary, which we may assume WLOG, we can separate these out by first choosing the random $A_i$'s and then subsampling to get the $R_i$'s.  So we want to prove that
    \[
        \mathbb{E}_{A_k, \dots, A_0} \left[\mathbb{E}_{R_k, \dots R_0} [|R_z|]\right] \leq 1 + 4\epsilon.
    \]
    Suppose we could prove \cref{lem:main-sampling} for a deterministic nonadaptive adversary, i.e., for a fixed $A_k, A_{k-1}, \dots, A_0$.  In other words, suppose that $\mathbb{E}_{R_k, \dots, R_0}{|R_z|} \leq 1 + 4\epsilon$ for all sequences $A_k, A_{k-1}, \dots, A_0$.  Then clearly 
    \[
        \mathbb{E}_{A_k, \dots, A_0} \left[ \mathbb{E}_{R_k, \dots R_0} [|R_z|]\right] \leq \mathbb{E}_{A_k, \dots, A_0} [1+4\epsilon] = 1 + 4\epsilon.
    \]
    Thus if we can prove~\cref{lem:main-sampling} against a nonadaptive deterministic adversary, we have proved~\cref{lem:main-sampling} against an adaptive and randomized adversary, as desired.
\end{proof}

So from now on, we may assume that the family $A_k, A_{k-1}, \dots, A_0$ is \textbf{fixed}.
Note that in this setting, the sets $R_i$ are independent of each other; this holds as the sets $A_i$ are fixed, and the randomness used to obtain $R_i$ is independent of the randomness used to sample other $R_j$ sets.  Before proving~\cref{lem:main-sampling}, we first show several auxiliary observations (all of which are in the setting where $A_k, A_{k-1}, \dots, A_0$ are fixed).

Our first observation is that for the first $b$ rounds of the \ssp, the expected number of sampled elements is small.
\begin{observation}\label{obs:low-initial}
Assume that $k \geq b \lceil \log_{1+\epsilon}(n_k / \epsilon)\rceil$.
Then, for each $j \in (k-b, k]$, $p_j \cdot n_j \leq \epsilon$.
\end{observation}

\begin{proof}
We have
\[
\lceil j / b \rceil = \lceil k / b \rceil \geq \log_{1+\epsilon}(n_k / \epsilon),
\]
which gives
\[
p_j \cdot n_j = (1+\epsilon)^{-\lceil j/b \rceil} \cdot n_j \leq (1+\epsilon)^{-\log_{1+\epsilon} (n_k / \epsilon)} \cdot n_k = \epsilon.
\]\end{proof}

We can also show that probabilities and the expected number of sampled elements do not increase much in any consecutive $b$ steps.
\begin{observation}\label{obs:slow-increase}
 For each $i \in [0, k)$, and any $j \in [i+1, i+b]$, $p_i \leq (1+\epsilon) p_j$ and $p_i \cdot n_i \leq (1+\epsilon) p_j\cdot n_j$.
\end{observation}
\begin{proof}
This first claim follows from the definition of $p_i$.
The second claim additionally uses the fact that $n_0, \ldots, n_k$ is a non-decreasing sequence.
\end{proof}

This observation now allows us to show that if we have a relatively large expected number of elements in $R_i$, then the probability that we have not yet sampled any elements in $R_j$ for $j > i$ is notably smaller than the probability that we haven't sampled any elements in $R_j$ for $j > i+b$.

\begin{lemma}\label{lem:reaching-probability}
Assume that $p_i \cdot n_i \geq \epsilon (1+\epsilon)$ for some $i \in [0, k-b]$.
Then, $P(z \leq i) \leq P(z \leq i+b) / (2+2\epsilon)$.
\end{lemma}
\begin{proof}
Denote by $E_j$ the event that $R_j = \emptyset$.
Recall that $z$ is the maximum index such that $R_z \neq \emptyset$.
Observe that the event that $z \leq x$ is equivalent to $\bigcap_{j > x} E_j$ and the individual events $E_j$ are independent.
Hence $P(z \leq i) = P(z \leq i + b)\cdot \prod_{j=i+1}^{i+b} P(E_j) $.
To complete the proof we will show that $\prod_{j=i+1}^{i+b} P(E_j) \leq 1/(2+2\epsilon)$.

By \cref{obs:slow-increase}, we have that for $j \in [i+1, \ldots, b]$, $(1+\epsilon) p_j \cdot n_j \geq p_i \cdot n_i \geq \epsilon (1+\epsilon)$,
which implies $p_j \cdot n_j \geq \epsilon$. Hence,
\[
P(E_j) = (1-p_j)^{n_j} \leq e^{-p_j \cdot n_j} \leq e^{-\epsilon}.
\]
By using the above, we get
\[
\prod_{j=i+1}^{i+b} P(E_j) \leq e^{-b \cdot \epsilon} = e^{-\lceil \log (2(1+\epsilon)) / \epsilon \rceil \epsilon} \leq e^{-\log(2(1+\epsilon))} = \frac1{2+2\epsilon}
\]
which finishes the proof.
\end{proof}

This now allows us to give an absolute bound on the probability that we have not sampled any elements before we sample $R_j$, assuming that we have a pretty high probability of sampling an element in $R_j$.
\begin{lemma}\label{lem:boundpini}
Let $j \in [0, k]$ be such that $p_j \cdot n_j \geq \epsilon (1+\epsilon)^c$ for a nonnegative integer $c$.
Then $P(z \leq j) \leq (2+2\epsilon)^{-c}$.
\end{lemma}

\begin{proof}
We prove the claim using induction on $c$.
For $c = 0$, $(2+2\epsilon)^{-c}$ is $1$, and the claim is trivially true.

Now, fix $c \geq 1$.
By \cref{obs:low-initial}, we have that $j \leq k-b$, and so we apply \cref{lem:reaching-probability} to obtain $P(z \leq j) \leq P(z \leq j+b) / (2+2\epsilon)$.
Hence, to complete the proof, it suffices to show $P(z \leq j+b) \leq (2+2\epsilon)^{-c+1}$.

We achieve that by applying the induction hypothesis to $j' = j+b$.
Indeed, by \cref{obs:slow-increase}, $p_{j'} \cdot n_{j'} \geq \epsilon(1+\epsilon)^{c-1}$, and so the assumptions of the inductive hypothesis hold.
As a result, we obtain $P(z \leq j+b) = P(z \leq j') \leq (2+2\epsilon)^{-c+1}$, as required.
\end{proof}

Before finally proving~\cref{lem:main-sampling}, we first show two more useful claims.

\begin{claim}\label{cl:condition-nonempty}
For each $i \in [0, k]$, it holds that $\E{|R_i| \mid R_i \neq \emptyset} \leq 1 + p_i \cdot n_i$.
\end{claim}

\begin{proof}
We have
$$\E{|R_i| \mid R_i \neq \emptyset} = \frac{\E{|R_i|}}{P(R_i \neq \emptyset)} = \frac{p_i\cdot n_i}{1 - (1-p_i)^{n_i}} \leq \frac{p_i\cdot n_i}{1 - \frac{1}{e^{p_i \cdot n_i}}} \leq \frac{p_i\cdot n_i}{1 - \frac{1}{1+p_i\cdot n_i}} = \frac{p_i \cdot n_i}{\frac{p_i \cdot n_i}{1+p_i \cdot n_i}} = 1 + p_i \cdot n_i.
$$
Note that in the first inequality we used the fact that $1-p_i \leq e^{-p_i}$, while in the second we used $e^{p_i\cdot n_i} \geq 1 + p_i \cdot n_i$.
\end{proof}

\begin{claim}\label{cl:pini}
Recall that $z$ is the maximum index such that $R_z \neq \emptyset$.
It holds that $\E{|R_z|} \leq 1 + \E{p_z \cdot n_z}$.
\end{claim}

\begin{proof}
Denote by $X_i$ the event that $R_j = \emptyset$ \emph{for all} $j \in [i+1, k]$.
Note that $z = i$ is the intersection of events $R_i \neq \emptyset$ and $X_i$.
\begin{align*}
\E{|R_z|} & = \sum_{i=0}^{k} \E{|R_i| \mid z = i} \cdot P(z = i) \\
& = \sum_{i=0}^{k} \E{|R_i| \mid R_i \neq \emptyset \cap X_i} \cdot P(z = i) = \sum_{i=0}^{k} \E{|R_i| \mid R_i \neq \emptyset} \cdot P(z = i)\\
& \leq \sum_{i=0}^{k} (1 + p_i \cdot n_i)\cdot P(z = i) = \sum_{i=0}^{k} P(z = i) + \sum_{i=0}^{k} p_i \cdot n_i \cdot P(z = i) = 1 + \E{p_z \cdot n_z},
\end{align*}
where the final inequality is from~\cref{cl:condition-nonempty}.  
We used the fact that $\E{|R_i| \mid R_i \neq \emptyset \cap X_i} = \E{|R_i| \mid R_i \neq \emptyset}$, which follows from the fact that $R_i$ is independent from $X_i$.
\end{proof}

We are now ready to prove~\cref{lem:main-sampling}.
\begin{proof}[Proof of \cref{lem:main-sampling}]
We know from \cref{obs:nonadaptive} that without loss of generality, the sequence $A_k, A_{k-1}, \dots, A_0$ is fixed.
By \cref{cl:pini} it suffices to show that $\E{p_z \cdot n_z} \leq 4 \epsilon$.
Let us define $X := p_z \cdot n_z$ to shorten notation.
\begin{align*}
\E{X} & \leq  P(X < \epsilon) \cdot \epsilon + \sum_{c=0}^\infty P\left(X \in \left[\epsilon(1+\epsilon)^c, \epsilon(1+\epsilon)^{c+1}\right)\right)\cdot \epsilon(1+\epsilon)^{c+1}\\
 & \leq  \epsilon + \sum_{c=0}^\infty P\left(X \geq \epsilon(1+\epsilon)^c\right)\cdot \epsilon(1+\epsilon)^{c+1}\\
 & \leq \epsilon + \sum_{c=0}^\infty 2^{-c}(1+\epsilon)^{-c} \cdot \epsilon (1+\epsilon)^{c+1}\\
 & \leq \epsilon (1 + 2(1+\epsilon)) \leq 4\epsilon.
\end{align*}
Note that we used the bound on $P(X \geq \epsilon(1+\epsilon)^c) \leq (2+2\epsilon)^{-c}$ which follows directly from \cref{lem:boundpini}.
\end{proof}

\subsection{Probability of the Sampled Element Being Not Unique}
We use the following lemma to analyze our \hm algorithm.
Specifically, it upper bounds the probability that an edge is sampled in \cref{alg:approx-matching}, but \emph{not} included in the final matching. In the proof of \cref{lemma:approx-matching}, we specify how to map our \hm algorithm to the setup in this section.

\singlesample*

\begin{proof}
As for the previous proof in this section, first assume that the sets $A_0, \ldots, A_k$ are fixed.

Our goal is to upper bound
\begin{align}
P(|R_z| > 1 \mid a \in R_z) & = \frac{P(|R_z| > 1 \cap a \in R_z)}{P(a \in R_z)}\label{l1}
\end{align}

Let $m_a = \min \{i \in [0, k] \mid a \in A_i\}$ be the index of the last step before $a$ is removed from the sets $A_0, \ldots, A_k$.
We obtain:

\begin{align*}
P(|R_z| > 1 \cap a \in R_z) & = \sum_{i=m_a}^k P(z = i \cap |R_i| > 1 \cap a \in R_i) & \\
& =  \sum_{i=m_a}^k P(z \leq i \cap  |R_i| > 1 \cap a \in R_i) & \textrm{since } a \in R_i \textrm{ implies  } z \geq i\\
& = \sum_{i=m_a}^k P(z \leq i) P(|R_i| > 1 \cap a \in R_i) & \substack{z \leq i \textrm{ is equivalent to } R_j = \emptyset \textrm{ for all } j > i\\\textrm{these events are independent of } R_i}
\end{align*}

Observe that the event $|R_i| > 1 \cap a \in R_i$ happens when $a$ is sampled and at least one out of the remaining $n_i - 1$ elements of $A_i$ are sampled.
Hence,
\[
P(|R_i| > 1 \cap a \in R_i) = p_i \cdot (1 - (1-p_i)^{n_i-1}) \leq p_i \cdot (1 - (1 - p_i \cdot (n_i-1))) \le p_i^2 \cdot n_i,
\]
and so we finally obtain $P(|R_z| > 1 \cap a \in R_z) \leq \sum_{i=m_a}^k P(z \leq i)p_i^2 \cdot n_i$. For the first inequality above, we used Bernoulli's inequality which states that $(1+x)^r\ge 1 + rx$ for every integer $r \ge 1$ and a real number $x \ge -1$.
Analogous reasoning allows us to show a similar identity for the denominator:
\[
P(a \in R_z) = \sum_{i=m_a}^k P(z \leq i) P(a \in R_i) =  \sum_{i=m_a}^k P(z \leq i) p_i
\]

Hence we can upper bound \cref{l1} as follows
\begin{align}\label{l2}
P(|R_z| > 1 \mid a \in R_z) \leq \frac{\sum_{i = m_a}^k P(z \leq i)P(a \in A_i)p_i^2 \cdot n_i}{\sum_{i = m_a}^k P(z \leq i)P(a \in A_i)p_i}.
\end{align}

\begin{claim}
Let $I = \{i \in [m_a, k] \mid p_i \cdot n_i < \epsilon (1+\epsilon)\}$ be a set of indices.
Then
$\sum_{i = m_a}^k P(z \leq i)p_i^2 \cdot n_i \leq 2 / (1-\epsilon) \sum_{i \in I} P(z \leq i)p_i^2 \cdot n_i$
\end{claim}

\begin{proof}
Consider the sum
\[
\sum_{i = m_a}^k P(z \leq i)p_i^2 \cdot n_i.
\]
We are going to charge the summands indexed by $[m_a,k] \setminus I$ to the summands indexed by $I$.
Formally, the charging is defined by a function $f : [m_a, k] \rightarrow [m_a,k]$.
We define $f(i)$ to be the smallest index $j \in \{i, i+b, i+2b, \ldots \}$ such that $j \in I$, that is $p_j \cdot n_j < \epsilon (1+\epsilon)$.
We note that $f(i)$ is well-defined since by \cref{obs:low-initial} we have that $(k-b, k] \subseteq I$.

Now we charge each summand $i$ to $f(i) \in I$ and show that the total charge of each summand in $I$ increases by at most a constant factor.
Let us now fix any $j \in I$ and consider the sum $\sum_{i \in f^{-1}(j)} P(z \leq i)p_i^2 \cdot n_i$.
Let $h := |f^{-1}(j)|$.
Then $f^{-1}(j) = \{j, j-b, \ldots, j-(h-1)b\}$.
We now show that the summands in the considered sum are geometrically decreasing (if we consider the indices in decreasing order).
Indeed, consider $x \in f^{-1}(j) \setminus \{j\}$.
We are now going to use the following facts.
\begin{itemize}
\item By \cref{lem:reaching-probability} we have $P(z \leq x) \leq P(z \leq x+b) / (2+2\epsilon)$.
\item By \cref{obs:slow-increase}, $p_x \cdot n_x \leq (1+\epsilon) p_{x+b} \cdot n_{x+b}$.
\item By \cref{obs:slow-increase}, $p_x \leq (1+\epsilon) p_{x+b}$.
\end{itemize}
These three facts together imply that for any $x \in f^{-1}(j) \setminus \{j\}$
\[
P(z \leq x)p_x^2 \cdot n_x \leq (1+\epsilon)/2\cdot P(z \leq x+b)p_{x+b}^2 \cdot n_{x+b}.
\]
Hence, the summands in $f^{-1}(j)$ can be arranged into a sequence in which the largest element is the summand corresponding to $j$, and each subsequent summand is at least a factor of $(1+\epsilon)/2$ smaller.
As a result, the total charge of the summand $j$ is $1 / (1-(1+\epsilon) / 2) = 2 / (1-\epsilon)$.
\end{proof}

Using the above claim, we upper bound \cref{l2}.

\begin{align*}
\frac{\sum_{i = m_a}^k P(z \leq i)p_i^2 \cdot n_i}{\sum_{i = m_a}^k P(z \leq i)p_i} & \leq  \frac{2 \cdot \sum_{i \in I} P(z \leq i)p_i^2 \cdot n_i}{(1-\epsilon) \cdot \sum_{i \in I} P(z \leq i)p_i}
 < \frac{2 \cdot \epsilon (1+\epsilon) \cdot \sum_{i \in I} P(z \leq i)p_i}{(1-\epsilon) \cdot \sum_{i \in I} P(z \leq i)p_i}
 \leq 6\epsilon.
\end{align*}

The proofs are stated while assuming that the sets $A_0, \ldots, A_k$ are fixed. As given by \cref{obs:nonadaptive}, this assumption can be made without loss of generality.
\end{proof}

\section{MPC Algorithms}
\label{sec:MPC-algorithms}
In this section, we show how to implement our algorithms in the MPC model.
Some of the approximation and total space results that we provide in this section are stated in expectation. Nevertheless, by applying Markov's inequality and the Chernoff bound, these can be turned into high probability guarantees by executing $O(\log n / \eps)$ instances of the algorithm in parallel and choosing the best output.
We present more details of this transformation in \cref{sec:expectation-to-whp} for completeness.

\begin{lemma}
\label{lemma:hypergraph-sparsification}
    Let $G = (V, E)$ be a rank $h$ hypergraph and $p \in [0, 1]$ a parameter. Let $H = (V, E_H)$ be a rank $h$ hypergraph, such that $E_H$ is obtained by taking each edge of $E$ independently with probability $p$. Let $n_H$ be the number of non-isolated vertices in $H$, and $\bh \le h$ the average edge-rank of $G$. Then,
    \[
        \E{n_H} \le p \cdot \bh \cdot |E|.
    \]
\end{lemma}
\begin{proof}
    Consider a vertex $v$ whose degree is $\deg(v)$. Let $X_v$ be a $0/1$ random variable being equal to $1$ if $v$ is non-isolated.
    Then, by Bernoulli's inequality,
    \begin{equation}
        \label{eq:expected-non-isolated}
        \prob{X_v = 1} \le 1 - (1 - p)^{\deg(v)} \stackrel{\text{Bernoulli's ineq.}}{\le} 1 - (1 - p \cdot \deg(v)) = p \cdot \deg(v).
    \end{equation}
    
    Observe that we also have $\sum_{v \in V} \deg(v) =  \bh \cdot |E|$. 
    From \cref{eq:expected-non-isolated} we further derive
    \[
        \E{n_H} = \sum_{v \in V} \E{X_v} = \sum_{v \in V} \prob{X_v = 1} \le \sum_{v \in V} p  \cdot \deg(v) = p \cdot \bh \cdot |E|. \qedhere
    \]
\end{proof}

%%%%%%%%%%%%
%%% New section
%%%%%%%%%%%%
\subsection{\cref{alg:f-apx} in the sub-linear space regime}
\label{sec:f-set-cover-sub-linear}
The analysis in this and \cref{sec:matching-sub-linear} (for~\cref{alg:approx-matching}) are almost identical -- they are presented in different languages, and different claims are used to conclude their approximation guarantees.
\begin{theorem}\label{thm:f-approx-MPC}
    Let $\eps \in (0, 1/2)$ be an absolute constant. There is an MPC algorithm that, in expectation, computes an $(1+\eps) f$ approximate set cover in the sub-linear space regime. 
    This algorithm runs in
    \[
        \tO\rb{\eps^{-1} \cdot \sqrt{\log \Delta} + \eps^{-2} \cdot  \log f + \eps^{-4} \cdot \log^2 \log n}
    \]
    MPC rounds, and each round, in expectation, uses total space of $O\rb{m}$.
\end{theorem}
\begin{proof}
    The MPC algorithm simulates \cref{alg:f-apx-fixed}, which executes $k = b \ceil{\log_{1 + \eps}\rb{\Delta / \eps}} = O\rb{\log (\Delta / \eps) / \eps^2}$ iterations, where $b$ is the constant defined in \cref{eq:define-p_i}.
        We split those $k$ iterations into multiple \emph{phases}, where Phase~$j$ consists of $r_j$ consecutive iterations of \cref{alg:f-apx-fixed}.
        We set the value of $r_j \le \log n$ in the rest of this proof.

    \paragraph*{Executing a single phase.}
    We leverage the idea of collecting the \emph{relevant} $r_j$-hop neighborhood around each element, i.e., a vertex $v \in T$ of the graph, which suffices to simulate the outcome of \cref{alg:f-apx-fixed} for $v$ over that phase.
    This relevant $r_j$-hop neighborhood is collected via graph exponentiation in $O(\log r_j)$ MPC rounds.
    Let the current phase correspond to iterations $J = \{i, i - 1, \ldots, i - r_j + 1\}$.
    Then, only elements in buckets $B_q$ for $q \in J$ are (potentially) relevant; those that appeared in the previous buckets or appear in the subsequent ones do not affect the state of $v$ in this phase.
    Note that when an element $e \in B_q$ is relevant, it does not necessarily mean that by iteration $q$ the element $e$ would still be unmarked, e.g., $e$ might have been marked in the iteration $q - 1 \in J$.

    Let $E_j$ be the set of unmarked elements at the beginning of phase $j$.
    We have that with probability at most $\min\cb{1, \sum_{q \in J} p_q} \le \min\cb{1, r_j \cdot p_{i - r_j + 1}}$ an element in $E_j$ is relevant for this phase.
    If the maximum set size concerning $E_j$ is $O\rb{\log n / p_{i}}$, then with a direct application of the Chernoff bound, we get that the maximum set size in $\bigcup_{q \in J} B_q$ is $O\rb{r_j \frac{p_{i - r_j + 1}}{p_{i}} \log n}$ with high probability. 
    Following our definition of $p_q$, see \cref{eq:define-p_i}, we have $O\rb{r_j \frac{p_{i - r_j + 1}}{p_{i}} \log n} = O\rb{r_j (1+\eps)^{r_j/b} \log n}$.
    Hence, with high probability, the number of relevant elements in a radius $r_j$ around a given set or an element is $O\rb{\rb{f r_j (1+\eps)^{r_j/b} \log n}^{r_j}} = O\rb{\rb{f (1+\eps)^{r_j/b} \log^2 n}^{r_j}}$, where we are using that $r_j \le \log n$.

    Recall that we assumed that the set size at the beginning of a phase is $O((1 / p_{i}) \cdot \log n)$. We now justify this assumption.
    Regardless of the assumption, a direct application of the Chernoff bound implies that whp each set $v$ that has $\Omega(1 / p_{i - r_j + 1} \cdot \log n)$ unmarked elements in $E'$ in iteration ${i - r_j + 1}$ will have at least one element in $B_{i - r_j + 1}$. 
    Therefore, $v$ will be removed at the end of this phase. This implies that the maximum number of unmarked elements in a set in the next phase is $O(1 / p_{i - r_j + 1} \cdot \log n)$, as we desire. Since the assumption holds in the very first phase by our choice of $p$, then the assumption on the maximum set size holds in each phase.

    \paragraph*{Linear total space.}
        Let $H$ be the subgraph of $G$ consisting of all the relevant elements of $E_j$ and the non-removed sets at the beginning of phase $j$ corresponding to iterations in $J$. 
        Now, we want to upper-bound the number of the non-zero degree sets and the elements in $H$. We then tie that upper bound to the size of $G$, providing the space budget per set/element.

        We can think of $H$ as a hypergraph in which an element is a hyper-edge connecting the sets it belongs to. Therefore, by \cref{lemma:hypergraph-sparsification}, the number of non-zero degree sets in $H$ is $O(r_j p_{i - r_j + 1} \bar{f} \cdot |E_j|)$, where $\bar{f}$ is the average element degree.
        The expected number of the relevant elements of $E_j$ is at most $r_j p_{i - r_j + 1} \cdot |E_j|$.

        Since the input size is $\Omega(\bar{f} \cdot |E_j|)$, our algorithm can assign $1/(r_j p_{i - r_j + 1})$ space to each non-empty set and relevant element in phase $j$ while assuring that in expectation the total used space is linear in the input size. Based on this, we compute the value of $r_j$.

        %For the sake of clarity, we let $Q = (1+\eps)^{r_i/b} \log^2 n$.
        From our discussion on the relevant $r_j$-hop neighborhood and the fact that each set and relevant element has the space budget of $1/(r_j p_{i - r_j + 1})$, our goal is to choose $r_j$ such that
        \begin{equation}\label{eq:condition-neighborhood}
            O\rb{\rb{f (1+\eps)^{r_j/b} \log^2 n}^{r_j}} \le \frac{1}{r_j p_{i - r_j + 1}},
        \end{equation}
        where $b$ is the constant defined in \cref{eq:define-p_i}.
        
        Our final round complexity will be tied to the drop of the maximum set degree. 
        So, it is convenient to rewrite the condition \cref{eq:condition-neighborhood}.
        Let $\tau_j$ be the upper bound on the maximum set degree at the beginning of phase~$j$. 
        Recall that we showed that with high probability $\tau_j = O(1 / p_{i} \cdot \log n)$.
        This enables us to rewrite \cref{eq:condition-neighborhood} as
        \begin{equation}\label{eq:condition-neighborhood-stronger}
            O\rb{\rb{f (1+\eps)^{r_j/b} \log^2 n}^{r_j}} \le \frac{\tau_j}{r_j (1+\eps)^{r_j/b} \log n},
        \end{equation}
        where we used that $p_{i - r_j + 1} = O((1 + \eps)^{r_j / b} \cdot p_{i})$. 
        To upper-bound $r_j$, we consider several cases.

        \begin{itemize}
            \item \textbf{Case 1: $\tau_j \le (\log n)^{O(\eps^{-2} \cdot \log \log n)}$.} In this case, we let $r_j = 1$. 
            \item \textbf{Case 2: $f \le (1+\eps)^{r_j/b} \log^2 n$ and Case~1 does not hold.}
            The condition in \cref{eq:condition-neighborhood-stronger} implies
            \begin{equation}\label{eq:condition-Case-2-first}
                O\rb{r_j (1+\eps)^{r_j/b} \log n \cdot \rb{f (1+\eps)^{r_j/b} \log^2 n}^{r_j}} \le \tau_j.
            \end{equation}
            Since $r_j \le \log n$, the leftmost $r_j$ in the condition above we replace by $\log n$, obtaining
            \begin{equation}\label{eq:condition-Case-2-second}
                O\rb{(1+\eps)^{r_j/b} \log^2 n \cdot \rb{f (1+\eps)^{r_j/b} \log^2 n}^{r_j}} \le \tau_j.                
            \end{equation}
            Note that, as long as $r_j \le \log n$, the upper-bound on $r_j$ derived from \cref{eq:condition-Case-2-second} is smaller/stricter than the one derived from \cref{eq:condition-Case-2-first}.
            Applying the same reasoning, we replace $f$ by $(1+\eps)^{r_j/b} \log^2 n$ deriving
            \begin{equation}\label{eq:condition-Case-2-third}
                O\rb{(1+\eps)^{r_j/b} \log^2 n \cdot \rb{(1+\eps)^{r_j/b} \log^2 n}^{2 r_j}} \le \tau_j.                
            \end{equation}
            Finally, by using that $r_j \ge 1$, we turn \cref{eq:condition-Case-2-third} into an event stricter upper-bound on $r_j$ by the following
            \[
                O\rb{\rb{(1+\eps)^{r_j/b} \log^2 n}^{3r_j}} \le \tau_j.
            \]
            By applying $\log$ on both sides and using that $\log (1 + \eps) = \Theta(\eps)$, this condition yields
            \[
                O\rb{r_j^2 \cdot \eps^{2} + r_j \cdot \log \log n} \le \log \tau_j.
            \]
            This implies two requirements: $r_j \le O\rb{\eps^{-1} \cdot \sqrt{\log \tau_j}}$ and $r_j \le O\rb{\log \tau_j / \log \log n}$. The former requirement is more stringent except for ``small'' values of $\tau_j$, i.e., except when $\eps^{-1} \cdot \sqrt{\log \tau_j} > \log \tau_j / \log \log n$. Hence, these ``small'' values of $\tau_j$ are when $\eps^{-1} \cdot \log \log n > \sqrt{\log \tau_j}$ holds. This implies $\tau_j < (\log n)^{\eps^{-2} \cdot \log \log n}$. 
            
            Since these values of $\tau_j$ are handled by Case~1, we conclude that Case~2 provides a condition $r_j = O\rb{\eps^{-1} \cdot \sqrt{\log \tau_j}}$.

            \item \textbf{Case 3: $f > (1+\eps)^{r_j/b} \log^2 n$ and Case~1 does not hold.}
                By re-arranging the terms, \cref{eq:condition-neighborhood-stronger} directly implies
                \[
                    r_j (1+\eps)^{r_j/b} \log n \cdot O\rb{\rb{f (1+\eps)^{r_j/b} \log^2 n}^{r_j}} \le \tau_j.
                \]
                Using that $r_j \le \log n$, the condition above is more constraining for $r_j$ after replacing the leftmost $r_j$ by $\log n$, which derives
                \begin{equation}\label{eq:Case-3-condition-r_j}
                    (1+\eps)^{r_j/b} \log^2 n \cdot O\rb{\rb{f (1+\eps)^{r_j/b} \log^2 n}^{r_j}} \le \tau_j.
                \end{equation}
                Using the lower-bound on $f$ that this case assume, from \cref{eq:Case-3-condition-r_j} we derive a stricter condition on $r_j$ which is
                \[
                    (1+\eps)^{r_j/b} \log^2 n \cdot O\rb{\rb{f^2}^{r_j}} \le \tau_j.
                \]
                Finally, this further yields
                \[
                    O\rb{f^{2 r_j + 1}} \le \tau_j,
                \]
                implying $r_j \le O\rb{\frac{\log \tau_j}{\log f}}$.
        \end{itemize}

        \paragraph*{Round complexity.}
            We now analyze the round complexity by building on the cases we studied above.
            
        \begin{itemize}
            \item \textbf{Case 1.} If $\tau_j \le (\log n)^{O(\eps^{-2} \cdot \log \log n)}$, and since $p_{i} = \Theta(\tau_j^{-1} / \log n)$, there are only $O(\eps^{-3} \cdot b \cdot \log^2 \log n)$ iterations remaining.
            Since, in this case, our MPC algorithm executes each iteration of \cref{alg:f-apx-fixed} in $O(1)$ MPC rounds, this case requires $O(\eps^{-3} \cdot b \cdot \log^2 \log n)$ MPC rounds.

            \item \textbf{Case 2.}
            To analyze Case~2, assume that $\tau_j \in [2^{2^t}, 2^{2^{t + 1}}]$ for an integer $t$. 
            Then, $\tau_{x} > 2^{2^{t-1}}$ holds for $O(\sqrt{b \eps^{-1}} \cdot 2^{t/2})$ different phases $x$ succeeding Phase $j$. This is the case as, from our analysis of Case~2, each of those Phases~$x$ simulates $\Omega(\sqrt{b \eps^{-1}} \cdot \sqrt{\log \tau_x}) = \Omega(\sqrt{b \eps^{-1}} \cdot 2^{(t-1) / 2})$ iterations of \cref{alg:f-apx-fixed}, where we use that $\sqrt{\log \tau_x} \ge \sqrt{\log 2^{2^{t-1}}} = \sqrt{ 2^{t-1}} = 2^{(t-1)/2}$. 
            So, all those phases simulate $\Omega(b \eps^{-1} \cdot 2^{t/2 + (t-1)/2}) = \Omega(b \eps^{-1} \cdot 2^{t + 1})$ iterations. After those many iterations, the upper bound on the maximum degree drops by a factor $2^{2^{t + 1}}$.
 
            Recall that each Phase $j$ is simulated in $O(\log r_j) = O(\log( b \cdot \eps^{-1} \cdot \log \tau_j)) = O(\log (b \cdot \eps^{-1} \cdot \log \Delta))$ MPC rounds.

            Before we provide the total round complexity, it remains to upper-bound $t$. 
            As a reminder, we set $\tau_j$ to $\Theta(1/p_i \cdot \log n)$, for appropriately chosen $p_i$.
            Hence, in Phase~$1$, we have that the corresponding $p_i$ equals $\Theta(\eps / \Delta)$, and hence $\tau_1 = c \Delta / \eps \cdot \log n$, for appropriately chosen constant $c \ge 1$. This leaves us with $t = \log\rb{\log (c \Delta / \eps) + \log \log n} \le \log(2 \log (c \Delta / \eps))$, where we are using that $\Delta \ge \log n$.
            
            Therefore, the overall round complexity when Case~2 holds is
            \[
                O(\log (b \cdot \eps^{-1} \cdot \log \Delta)) \sum_{t = 0}^{\log (2 \log (c\Delta / \eps))} \sqrt{b \cdot \eps^{-1}} \cdot 2^{t/2} = O\rb{\log (b \cdot \eps^{-1} \cdot \log \Delta) \cdot \sqrt{b \cdot \eps^{-1}} \cdot \sqrt{\log(\Delta / \eps)}}.
            \]

            \item \textbf{Case 3.}
                The round complexity analysis for this case is almost identical to Case~2. 
                In Case~3, our algorithm simulates $\Theta\rb{\frac{\log \tau_j}{\log f}}$ iterations in Phase~$j$.
                As in Case~2, assume that $\tau_j \in [2^{2^t}, 2^{2^{t + 1}}]$ for an integer $t$. 
                Then, using the same reasoning as in Case~2, $\tau_{x} > 2^{2^{t-1}}$ holds for $O(b \eps^{-1} \cdot \log f)$ different phases $x$ succeeding Phase $j$.

                Hence, the round complexity in this case is
                \[
                    O(\log (\log \Delta / \log f) \sum_{t = 0}^{\log (2 \log (c\Delta / \eps))} b \cdot \eps^{-1} \cdot \log f = O\rb{\log^2 \log (\Delta/\eps) \cdot b \cdot \eps^{-1} \cdot \log f}.
                \]
        \end{itemize}

    \paragraph*{Approximation analysis.} Since we fully simulate \cref{alg:f-apx}, our MPC algorithm is also a $(1+\eps) f$-approximate one.
\end{proof}

%%%%%%%%%%%%
%%% New section
%%%%%%%%%%%%
\subsection{\cref{alg:log-apx} in the nearly-linear space regime}
This section shows that \cref{alg:log-apx} can be simulated in $O(\log \Delta)$ MPC rounds in the nearly-linear space regime.

\begin{theorem}\label{thm:MPC-log-approx-linear}
    Let $\eps \in (0, 1/2)$ be an absolute constant. There is an MPC algorithm that, in expectation, computes an $(1+\eps) H_{\Delta}$ approximate set cover with $\tO(|S| + |T|^\delta)$ space per machine, where $\delta > 0$ is an arbitrary constant. This algorithm succeeds with high probability, runs in $O\rb{\log \Delta}$ MPC rounds, and uses a total space of $\tO\rb{m}$.
\end{theorem}
Our MPC algorithm performs round compression of the inner loop of \cref{alg:log-apx}; that MPC approach is given as \cref{alg:MPC-set-cover-log-approx}.
Each iteration of the outer loop of \cref{alg:log-apx} is executed by invoking \cref{alg:MPC-set-cover-log-approx}.
The main part of round compression, i.e., \cref{alg:MPC-set-cover-log-approx}, is performed on a single machine. In one round of compression, and for a \emph{fixed} $j$ of the for-loop on \cref{line:log-n-alg-loop-j}, our goal is to execute the entire for-loop on \cref{line:log-n-alg-loop-i}. 
\begin{algorithm}[H]
\caption{MPC simulation of the inner-loop of \cref{alg:log-apx} in the nearly-linear space regime}\label{alg:MPC-set-cover-log-approx}
\begin{algorithmic}[1]

\Function{MPC-SetCover}{$G, \epsilon, j$}\Comment{$G=(S \cup T, E)$}
    \State Let $k = b \ceil{\log_{1+\epsilon} (f / \epsilon)}$ \label{line:define-A}
    \State Let $q_j \eqdef \min \cb{\frac{100}{\eps^2} \cdot \frac{\log n}{(1+\eps)^j}, 1}$ \label{line:define-q_j}

    \State Prepare $k+1$ independent samples $X_0, \ldots, X_k$ as follows. For each $i = 0 \ldots k$, sample $X_i$ is a subset of $T$ such that each element in $T$ is included in $X_i$ with probability $q_j$ and independently of other elements.

    \State $C \gets \emptyset$

    \For{$i = k$ \textbf{down to} $0$}
        \State Use sample $X_{i}$ to estimate the current degree $\tN(s)$ of each element in $S$. That is, let $\tN(s)$ be the \emph{minimum} of $\frac{|S \cap X_t|}{q_j}$ and, if $i < A$, the value of $\tN(s)$ in the iteration $i + 1$. \label{line:MPC-setcover-estimate-degree}
        \State{$D \gets \sample\rb{\cb{s \in S \mid \tN(s) \geq (1+\epsilon)^j}, p_i}$ \label{line:MPC-log-n-sample-D}}
        \State{$C \gets C \cup D$}
        \State{Remove from $G$ all sets in $D$ and all elements they cover}
    \EndFor
    \State \Return{C}
\EndFunction

\end{algorithmic}
\end{algorithm}

The for-loop of \cref{alg:MPC-set-cover-log-approx} is executed on a single machine, i.e., all the sets $S$, all the samples $X_i$, and all the edges between the sets and the samples are placed on one machine. Observe that \cref{line:MPC-setcover-estimate-degree} takes $\tN(s)$ to be the minimum estimate across all iterations $t \ge i$. This ensures that if $\tN(s)$ does not pass the threshold $(1+\eps)^j$ once, it will remain below it. In particular, such property will be important in our proof of \cref{thm:log-approx-MPC-sub-linear}.

\paragraph*{Space requirement.}
To analyze the space requirement, we first upper bound the degree of sets as a function of $j$.
\begin{lemma}\label{lemma:set-size-after-iteration-j}
    Consider an execution of \cref{alg:MPC-set-cover-log-approx} with parameter $G, \eps$ and $j$. After this execution, with high probability, the degree of each remaining set is $(1+\eps) \cdot (1+\eps)^j$.
\end{lemma}
\begin{proof}
    In \cref{alg:MPC-set-cover-log-approx}, consider the iteration $i = 0$. In that case, by \cref{eq:define-p_i}, it holds $p_0 = 1$. Hence, each set $s$ for which it holds $\tN(s) \ge (1+\eps)^j$ is included in $D$ on \cref{line:MPC-log-n-sample-D} and afterward removed from $G$.
    If $q_j = 1$, then $\tN(s) = N(s)$, and hence a set $s$ such that $|N(s)| \ge (1 + \eps)^j$ is removed from graph. So, assume $q_j = \frac{100}{\eps^2} \cdot \frac{\log n}{(1+\eps)^j}$.

    Consider a set that at the beginning of the iteration $i = 0$ has at least $(1+\eps) \cdot (1+\eps)^j$ many elements. We derive
    \[
        \E{|s \cap X_i|} = |s| \cdot q_j \ge (1+\eps) \cdot (1+\eps)^j \cdot q_j = \frac{100}{\eps^2} \cdot (1+\eps) \cdot \log n.
    \]
    Observe that for $\eps \in (0, 1)$ it holds that $(1 - \eps / 2) \cdot \E{|s \cap X_i|} \ge \frac{100}{\eps^2} \cdot \log n$.
    Hence, by the Chernoff bound (\cref{lemma:chernoff}~\cref{item:delta-at-most-1-le} for $\delta = \eps/2$), we have
    \[
        \prob{|s \cap X_i| \le \frac{100}{\eps^2} \cdot \log n} < n^{-10}.
    \]
    Hence, with probability $1 - n^{-10}$ at least, it holds that $\tN(s) \ge (1+\eps)^j$ and therefore $s$ is included in $D$ on \cref{line:MPC-log-n-sample-D}.

    By taking the Union bound over all the sets, which there are at most $n$, we get that, after the execution of \cref{alg:MPC-set-cover-log-approx} with parameter $G, \eps$ and $j$, with probability at least $1 - n^{-9}$ each non-removed set has degree $(1+\eps) \cdot (1 + \eps)^j$.
\end{proof}

\begin{lemma}\label{lemma:MPC-log-space}
    Let $k$ be defined on \cref{line:define-A} of \cref{alg:MPC-set-cover-log-approx}. Then, with high probability, \cref{alg:MPC-set-cover-log-approx} uses $O\rb{|S| \cdot \log n \cdot k / \eps^2}$ space.
\end{lemma}
\begin{proof}
Recall that each element in $X_i$ is sampled with probability $q_j$, defined on \cref{line:define-q_j}. Hence, in expectation and assuming that each set $s$ has size $O\rb{(1+\eps)^j}$, $\E{|s \cap X_k|} = O(\log n / \eps^2)$. Moreover, by the Chernoff bound, we have that $|s \cap X_i| = O(\log n / \eps^2)$ holds with high probability.

The claim follows since there are $k + 1$ samples $X_i$.
\end{proof}

Note that \cref{thm:MPC-log-approx-linear} also requires that each machine has at least $|T|^{\delta}$ space for an arbitrary constant $\delta > 0$. As \cref{lemma:MPC-log-space} shows, this is not necessary to execute \cref{alg:MPC-set-cover-log-approx} itself, but it is required by the standard MPC primitives~\cite{goodrich2011sorting} that would fetch data that \cref{alg:MPC-set-cover-log-approx} needs.

\paragraph*{Approximation analysis.}
To analyze the approximation guarantee of \cref{alg:MPC-set-cover-log-approx}, we re-use the proof of \cref{lemma:set-cover-log-approx}. The only properties of \cref{alg:log-apx} that \cref{lemma:set-cover-log-approx} relies on are those stated by \cref{lem:logn-properties}. So, we prove such properties for \cref{alg:MPC-set-cover-log-approx} as well.

\begin{lemma}\label{lem:logn-properties-MPC}
\cref{alg:MPC-set-cover-log-approx} adds sets to the solution in batches.
When a batch of sets $D$ is added to the solution, we have that with high probability (a) the residual size of each set in $D$ is at most $(1+3\epsilon)$ smaller than the maximum residual size of any set at that moment, and (b) for each newly covered element $t$, the expected number of sets in a batch that cover it is at most $(1+12\epsilon)$.
\end{lemma}
Note that when \cref{lem:logn-properties-MPC} and \cref{lem:logn-properties} are compared, then the former claim holds with high probability. It implies that our approximation guarantee for \cref{alg:MPC-set-cover-log-approx} also holds with high probability.
\begin{proof}[Proof of \cref{lem:logn-properties-MPC}]
    To prove Property~(a), we want to show that the ratio of the sizes of the sets sampled in $D$ is ``close'' to $1$. \cref{lemma:set-size-after-iteration-j} already upper bounds the largest set added to $D$; w.h.p.\ it has size $(1+\eps)^{j+1}$ at most. By carrying an almost identical analysis to the one in the proof of \cref{lemma:set-size-after-iteration-j}, we show that whp $D$ does not contain a set with a size smaller than $(1+\eps)^{j-1}$. To that end, consider a set $s$ of size $(1+\eps)^{j-1}$ at most. We have
    \[
        \E{|s \cap X_i|} \le \frac{100}{\eps^2} \cdot (1+\eps)^{-1} \cdot \log n.
    \]
    Observe that for $\eps \in (0, 1)$ it holds that $(1+\eps/2) / (1+\eps) < 1$ and hence $(1 + \eps / 2) \cdot \E{|s \cap X_i|} < \frac{100}{\eps^2} \cdot \log n$
    Therefore, by the Chernoff bound (\cref{lemma:chernoff}~\cref{item:delta-at-most-1-ge} for $\delta = \eps/2$), we have
    \[
        \prob{|s \cap X_i| \ge \frac{100}{\eps^2} \cdot \log n} < n^{-8}.
    \]
    So, w.h.p.\ the ratio of the set sizes sets included in $D$ on \cref{line:MPC-log-n-sample-D} is at most $(1+\eps)^2 \le 1 + 3\eps$.

    The proof of Property~(b), just as in \cref{lem:logn-properties}, follows from \cref{lem:main-sampling} when $\eps$ is replaced by $3\eps$.
\end{proof}

Coupling \cref{lem:logn-properties-MPC} with the proof for \cref{lemma:set-cover-log-approx} yields our desired approximation guarantee for \cref{alg:MPC-set-cover-log-approx}.
\begin{lemma}\label{lemma:approximation-MPC-log}
    With probability $1 - 1/n$ at least, \cref{alg:MPC-set-cover-log-approx} in expectation computes a $(1+O(\eps)) H_{\Delta}$-approximate solution to \setcover.
\end{lemma}
\cref{lemma:approximation-MPC-log,lemma:MPC-log-space} yield the approximation guarantee and the round complexity stated by \cref{thm:MPC-log-approx-linear}. To conclude the round complexity analysis, recall that we designed \cref{alg:MPC-set-cover-log-approx} to replace the entire inner loop of \cref{alg:log-apx}. 
Hence, the complete MPC algorithm corresponding to \cref{thm:MPC-log-approx-linear} executes \cref{alg:MPC-set-cover-log-approx} for $j = \floor{\log_{1+\eps} \Delta}$ down to $0$, where those invocation correspond to the outer loop of \cref{alg:log-apx}. This yields the advertised $O(\log \Delta)$ round complexity.

%%%%%%%%%%%%%%%%%%%
% New section
%%%%%%%%%%%%%%%%%%%
\subsection{\cref{alg:log-apx} in the sub-linear space regime}
In this section, we describe how to simulate \cref{alg:log-apx} in $O_\eps\rb{\log{\Delta} \cdot (\sqrt{\log f} + \log^2 \log n)}$ rounds in the sub-linear space regime. To that end, we follow the same graph-exponentiation recipe as the one we used to prove \cref{thm:f-approx-MPC}, but this time the graph-exponentiation is applied to \cref{alg:MPC-set-cover-log-approx}. As such, the approximation guarantee directly carries over from \cref{lemma:approximation-MPC-log}.
\begin{theorem}\label{thm:log-approx-MPC-sub-linear}
    Let $\eps \in (0, 1/2)$ be an absolute constant. There is an MPC algorithm that, in expectation, computes an $(1+\eps) H_{\Delta}$ approximate set cover in the sub-linear space regime. 
    This algorithm runs in
    \[
        \tO\rb{\log \Delta \cdot \rb{\eps^{-1} \cdot \sqrt{\log f} + \eps^{-4} \cdot \log^2 \log n}}
    \]
    MPC rounds, and each round, in expectation, uses total space of $O\rb{m}$.
\end{theorem}
\begin{proof}
    Our proof consists of simulating \cref{alg:MPC-set-cover-log-approx} in $\tO_\eps\rb{\sqrt{\log f}}$ rounds as follows. Recall that \cref{alg:MPC-set-cover-log-approx} is executed in place of the inner loop of \cref{alg:log-apx}.
    \cref{alg:MPC-set-cover-log-approx} executes $k = b \ceil{\log_{1 + \eps}\rb{f / \eps}}$ iterations, where $b$ is the constant defined in \cref{eq:define-p_i}.
    We split those $k$ iterations into multiple \emph{phases}, where Phase~$j$ consists of $r_j$ consecutive iterations of \cref{alg:f-apx}.
    We set the value of $r_j \le \log n$ in the rest of this proof.    

    \paragraph*{Executing a single phase.}
    We follow the idea of collecting the \emph{relevant} $r_j$-hop neighborhood around each element and set, i.e., a vertex $v$ of the graph, which suffices to simulate the outcome of \cref{alg:MPC-set-cover-log-approx} for $v$ over that phase.
    This relevant $r_j$-hop neighborhood is collected via graph exponentiation in $O(\log r_j)$ MPC rounds.
    Let the current phase correspond to iterations $J = \{i, i + 1, \ldots, i - r_j + 1\}$.
    Then, only those sets sampled during the iterations $J$ and only those elements in $X_k$ sets are (potentially) relevant; those that got removed in the previous phases or would be sampled in the subsequent phases do not affect the state of $v$ in this phase.
    So, to determine whether a set $s$ is potentially relevant for this phase, we first draw randomness for $s$ for $\sample$ concerning probability $p_q$ for each $q \in J$. 
    If any of those drawn random bits shows that $s$ should be in the corresponding $\sample$, we consider $s$ as relevant. 
    We remark that this process simply fixes the randomness used by $\sample$. 
    Also, if it turns out that $s$ is relevant because it \emph{would be} sampled by $\sample$ for $p_q$, it does not mean that by iteration $q$ the set $s$ would still have the appropriate size.

    Next, we want to upper-bound the size of $r_j$-hop neighborhood of each relevant element and set. 
    First, recall that by our proof of \cref{lemma:MPC-log-space}, with high probability, a set has $O(\log n \cdot k / \eps^2) = O\rb{\log^2 n / \eps^3}$ elements in $\cup_i X_i$. (In the graph exponentiation, a set exponentiates only through those elements in its $\cup_i X_i$.)
    This upper-bounds the ``relevant'' degree of a set. Second, we upper-bound the ``relevant'' degree of an element.
    Let $S_j$ be the family of sets at the beginning of a given phase whose estimated sizes still did not fall below the threshold given by \cref{line:MPC-setcover-estimate-degree,line:MPC-log-n-sample-D} of \cref{alg:MPC-set-cover-log-approx}.
    With probability at most $\min\cb{1, \sum_{q \in J} p_q} \le \min\cb{1, r_j \cdot p_{i - r_j + 1}}$ a set in $S_j$ is relevant for this phase.
    If the maximum element degree in $S_j$ is $O\rb{\log n / p_i}$, then with a direct application of the Chernoff bound, we get that the maximum element degree in $\bigcup_{q \in J} \sample(S_j, p_q)$ is $O\rb{r_j \frac{p_{i - r_j + 1}}{p_{i}} \log n}$ with high probability. 
    Following our definition of $p_j$, see \cref{eq:define-p_i}, we have $O\rb{r_j \frac{p_{i - r_j + 1}}{p_{i}} \log n} = O\rb{r_j (1+\eps)^{r_j/b} \log n}$.

    Combining these two upper-bounds, whp, the number of relevant sets and elements in the $r_j$-hop neighborhood of a given set or an element is $O\rb{\rb{r_j (1+\eps)^{r_j/b} \log^3 n / \eps^3}^{r_j}}$, which by using $r_j \le \log n$ simplifies to $O\rb{\rb{(1+\eps)^{r_j/b} \log^4 n / \eps^3}^{r_j}}$.

    Recall that we assumed that the relevant degree of an element at the beginning of a phase is $O(1 / p_{i} \cdot \log n)$. We now justify this assumption.
    Regardless of the assumption, a direct application of the Chernoff bound implies that, with high probability, each element $e$ that has $\Omega(1 / p_{i - r_j + 1} \cdot \log n)$ neighboring sets in $S_j$, such that those sets pass the size threshold at the end of that phase, will have at least one element in $\sample(S_j, p_{i - r_j + 1})$.
    Therefore, $e$ will be removed at the end of this phase. We remark that our size estimates on \cref{line:MPC-setcover-estimate-degree} of \cref{alg:MPC-set-cover-log-approx} result in a monotonic element-degree with respect to $S_j$ in subsequent phases. This is helpful because once the relevant element degree falls below $\log n / p_{i-r_j+1}$, it will remain below that threshold for a fixed $j$.
    This implies that the maximum element degree in the next phase is $O(\log n / p_{i - r_j + 1})$, as we desire. 
    Since the assumption holds in the very first phase by our choice of $p$, then the assumption on the maximum set size holds in each phase.

    \paragraph*{Linear total space.}
        Let $H$ be the subgraph of $G$ consisting of all the relevant sets of $S_j$ and the non-removed elements at the beginning of phase~$j$ corresponding to iterations in $J$.
        Now, we want to upper-bound the number of the non-zero degree elements and the sets in $H$. 
        We then tie that upper bound to the size of $G$, providing the space budget per set/element.

        The expected number of relevant sets is $O(r_j p_{i - r_j + 1} \cdot |S_j|)$. By the proof of \cref{lem:logn-properties-MPC}, the sizes of the sets in $S_j$ at the beginning of phase~$j$ are $\Theta(\Delta_j)$, for some $\Delta_j$. 
        Therefore, $G$ has size $\Omega\rb{\Delta_j \cdot |S_j|}$. This implies we have at least $1/(r_j p_{i - r_j + 1})$ space budget per relevant set.
        
        To provide the space budget per non-zero degree element in $H$, we think of $H$ as a hypergraph in which a set is a hyper-edge connecting the elements it contains.
        Therefore, by \cref{lemma:hypergraph-sparsification}, the number of non-zero degree sets in $H$ is $O(r_j p_{i - r_j + 1} \cdot \Delta_j \cdot |S_j|)$.

        Since $G$ has size $\Omega(\Delta_j \cdot |S_j|)$, our algorithm can assign $1/(r_j p_{i - r_j + 1})$ space to each relevant set and non-zero degree element in phase $j$ while assuring that in expectation the total used space is linear in the input size. 
        Based on this, we compute the value of $r_j$.

        %For the sake of clarity, we let $Q = (1+\eps)^{r_i/b} \log^2 n$.
        From our discussion on the relevant $r_j$-hop neighborhood and the fact that each relevant set and element has the space budget of $1/(r_j p_{i - r_j + 1})$, our goal is to choose $r_j$ such that
        \begin{equation}\label{eq:condition-neighborhood-log-delta}
            O\rb{\rb{(1+\eps)^{r_j/b} \frac{\log^4 n}{\eps^3}}^{r_j}} \le \frac{1}{r_j p_{i - r_j + 1}},
        \end{equation}
        where $b$ is the constant defined in \cref{eq:define-p_i}.
        
        Our final round complexity will be tied to the drop of the maximum element degree. 
        So, it is convenient to rewrite the condition \cref{eq:condition-neighborhood-log-delta}.
        Let $\tau_j$ be the upper bound on the maximum element degree at the beginning of phase $j$. 
        Recall that we showed that with high probability $\tau_j = O(1 / p_{i} \cdot \log n)$.
        This enables us to rewrite \cref{eq:condition-neighborhood-log-delta} as
        \begin{equation}\label{eq:condition-neighborhood-stronger-log-delta}
            O\rb{\rb{(1+\eps)^{r_j/b} \frac{\log^4 n}{\eps^3}}^{r_j}} \le \frac{\tau_j}{r_j (1+\eps)^{r_j/b} \log n},
        \end{equation}
        where we used that $p_{i - r_j + 1} = O((1 + \eps)^{r_j / b} \cdot p_{i})$. 

        The analysis of this is now identical to that presented in the proof of \cref{thm:f-approx-MPC} for $f = \frac{\log^2 n}{\eps^3}$.
\end{proof}

%%%%%%%%%%%%%%%%%%%
% New section
%%%%%%%%%%%%%%%%%%%
\subsection{\cref{alg:approx-matching} in the sub-linear space regime}
\label{sec:matching-sub-linear}
We observe that \cref{alg:approx-matching} is \cref{alg:f-apx} with an additional step of removing non-isolated edges, which can be easily implemented in $O(1)$ MPC rounds.
Hence, all but the approximation analysis we provided in \cref{sec:f-set-cover-sub-linear} carries over to \cref{alg:approx-matching}. The main difference is that the approximation guarantee in \cref{sec:f-set-cover-sub-linear} uses \cref{lem:main-sampling} while, in this section, we use \cref{lemma:approx-matching}. Recall that \cref{lemma:approx-matching} states that \cref{alg:approx-matching} outputs a $(1 - \eps h) / h$ approximation. Since our goal is to output a $(1 - \eps') / h$ approximation, we invoke \cref{alg:approx-matching} with $\eps = \eps' / h$. This, together with the analysis given in \cref{sec:f-set-cover-sub-linear}, yields the following claim. (The parameter $h$ in \cref{thm:hypermatching-MPC} corresponds to $f$ in \cref{thm:f-approx-MPC}.)

\begin{theorem}
\label{thm:hypermatching-MPC}
    Let $\eps \in (0, 1/2)$ be an absolute constant. There is an MPC algorithm that, in expectation, computes a $(1-\eps) / h$ approximate maximum matching in a rank $h$ hypergraph in the sub-linear space regime.
    This algorithm runs in
    \[
        \tO\rb{h \eps^{-1} \cdot \sqrt{\log \Delta} + h^2 \eps^{-2} + h^4 \eps^{-4} \cdot \log^2 \log n}
    \]
    MPC rounds, and each round, in expectation, uses total space of $O\rb{m}$.
\end{theorem}

\subsection{From expectation to high probability}
\label{sec:expectation-to-whp}
Some of the results in this section are presented so that the approximation and the total space usage are given in expectation.
This guarantee can be translated to ``with high probability'' using standard techniques, i.e., execute $O(\log n / \eps)$ instances in parallel, and, among the ones that do not overuse the total space by a lot, output the one with the best objective.
We now provide details.

Each of the algorithms outputs a $1+\eps$ approximation in expectation.
By Markov's inequality, with probability at most $1/(1+\eps) \le 1 - \eps/2$, the corresponding approximation is worse than $(1 + \eps)^2 \le 1 + 3\eps$.

Similarly, an algorithm uses $O(m)$ total space in expectation in each of the $t \le \poly \log n$ rounds. Hence, by Markov's inequality again, it uses more than $O(3 \cdot t m / \eps)$ total space in any round with probability $\eps / 3$ at most. Hence, with probability at least $\eps/2 - \eps/3 = \eps/6$ an algorithm both outputs $1+3\eps$ approximation and uses at most $\tO(m/\eps)$ total space.

Hence, with a high probability, at least one of $O(\log n / \eps)$ executed instances has all desired properties.

\section{PRAM Algorithms}\label{sec:pram}

Next, we discuss how to implement \cref{alg:f-apx} and \cref{alg:log-apx} in the concurrent-read concurrent write (CRCW) PRAM model (see \cref{sec:prelim}).
Most steps in our algorithms work in the weaker EREW model, with the only steps requiring concurrent reads or writes being those that broadcast coverage information from elements to sets (or vice versa).

\subsection{\cref{alg:f-apx} in the CRCW PRAM}

\begin{theorem}
\label{thm:f-approx-PRAM}
Let $\epsilon \in (0, 1/2)$ be an absolute constant. There is an algorithm that, in expectation, computes an $(1+\epsilon)f$-approximate set cover, where each element appears in at most $f$ sets. The algorithm deterministically runs in $O(n + m)$ work and $O(\log n)$ depth on the CRCW PRAM.
\end{theorem}

We first discuss how to fix the random choices made by our algorithms upfront in the PRAM model.
We note the Alias method can be parallelized~\cite{hubschle2022parallel}, although methods with depth linear in the number of steps ($k$) suffice for our algorithms since our PRAM implementations use at least unit depth per step and there are $k$ steps total.
Since the Alias method summarized in \cref{lem:aliasmethod} runs in $O(k)$ work and depth, we can use any efficient sequential Alias method algorithm, e.g., that of Walker~\cite{walker1974new} or Vose~\cite{vose1991linear} in our parallel algorithms.
The output is a query data structure that can then be queried in $O(1)$ work and depth to sample $X_t$, the round when element $t \in T$ will be sampled.

Our PRAM implementation of \cref{alg:f-apx} works as follows.
We start by precomputing the $X_t$ values, which can be done using $O(n)$ work and $O(k)$ depth.
We then use parallel integer sorting to group all elements processed in a particular round $i$ together.
The sorting can be done in $O(n)$ work and $O(\log n + k)$ depth using a parallel counting sort~\cite{rajasekaran1989optimal, blelloch23notes}.
We maintain an array of flags $B$ where $B_t = \mathsf{true}$ implies that element $t$ is covered.

In the $i$-th round of the algorithm, we process the set $D_i$ of all elements that are sampled in round $i$, which are stored consecutively after performing the integer sort.
Some elements in $D_i$ may be already covered, so we first check $B$ for each element in $D_i$ and mark any elements that are already covered.
For each element $t$ in $D_i$ that is not yet covered, we process all sets containing $t$ (none of which are already chosen) and mark this set as chosen by setting a bit for this set using concurrent writes.
We then map over all of the sets chosen in this round and mark all of their incident elements as covered in $B$ using concurrent writes.

The $i$-th round can therefore be implemented in $O(|D_i| + \sum_{t \in D_i} |N(t)|)$ work and $O(1)$ depth.
Summing over all steps, the work is $O(n + m)$ and the depth is $O(\log \Delta)$ since there are $k = b\lceil \log_{1+\epsilon}(\Delta / \epsilon)\rceil \in O(\log \Delta)$ steps total.
Overall, the depth of the algorithm is $O(\log n + \log \Delta) = O(\log n)$.

\subsection{\cref{alg:log-apx} in the CRCW PRAM}

Next, we show how to efficiently implement \cref{alg:log-apx} in the CRCW PRAM. 
Compared with other bucketing-based \setcover algorithms~\cite{berger1994efficient, blelloch2011linear}, our algorithm achieves work-efficiency and low depth while being significantly simpler.

\begin{theorem}\label{thm:PRAM-log-approx-setcover}
Let $\epsilon \in (0, 1/2)$ be an absolute constant. There is an algorithm that, in expectation, computes an $(1+\epsilon)\hdelta$-approximate set cover. 
The algorithm deterministically runs in $O(n + m)$ work and $O(\log^2 n \log\log n)$ depth on the CRCW PRAM.
\end{theorem}

Call each iteration of the loop on \cref{line:log-n-alg-loop-j} as a round, and each iteration of the inner loop on \cref{line:log-n-alg-loop-i} a step.
We first bucket the sets based on their initial degrees using a parallel counting sort to determine the first round in which they will be considered.
Let $S_j$ denote the set of sets that start in the $j$-th round, i.e., sets that have size $|N(s)| \geq (1+\epsilon)^{j}$.

Our algorithm starts a round $j$ by first {\em packing out} the sets in $S_j$, the set of sets considered in this round.
Packing is necessary for work-efficiency, as the algorithm may process the same set across multiple rounds, and repeatedly processing both uncovered and already covered elements in each round will lead to $O(m \log n)$ work.
The {\em pack} operation removes (or packs out) all already covered elements for a set, which can be implemented in $O(|N(s)|)$ work and $O(\log |N(s)|)$ depth for a set $s$ using parallel prefix sum.

\paragraph*{Size Estimation.}
Consider a round $j$, and the set of sets $S_j$ considered in this round. 
The round needs to run a sequence of steps that sample all sets in this round that have size at least $(1+\epsilon)^{j}$.
The difficulty is that after some number of steps, a set $s \in S_j$ can drop below the size constraint due to some of their elements becoming covered by other sets.
Our implementation, therefore, needs some way to compute (or estimate) the size of a set as the elements are gradually covered in very low depth (ideally $O(1)$ depth).
We need to implement the size estimation accurately enough, e.g., within a $(1+\delta)$-factor; otherwise, we risk sacrificing accuracy.
We also need to rebucket sets that are kicked out of this round to later rounds, which we will deal with later.
Note that this size-estimation problem was not addressed by earlier bucketing-based \setcover algorithms~\cite{berger1994efficient, blelloch2011linear}, since the other operations applied in an inner step (e.g., an iteration of \manis{}) such as packing or filtering a graph representation of the \setcover instance already require logarithmic depth. 

We show how to accurately estimate the size of each set in $O(\log \log n)$ depth.
If the sets all have size $O(\log n)$, we can simply sum over their incident elements in $O(\log \log n)$ depth.
If the sets are larger, computing the exact sum is infeasible due to known lower bounds on EREW PRAMs stating that any algorithm requires $\Omega(\log n / \log\log n)$ depth to solve majority~\cite{beame1989optimal}.
Instead, we use a parallel approximate prefix-sum algorithm~\cite{goldberg1995optimal} summarized in the following lemma:
\begin{lemma}[\cite{goldberg1995optimal}]
\label{lem:sizeest}
Given a sequence of integers $A = a_1, a_2, \ldots, a_n$, a sequence $B = b_1, b_2, \ldots, b_i$ is an approximate prefix-sum of $A$ if for all $j \in [n]$, $\sum_{j=1}^{i} a_j \leq b_i \leq (1+\delta)\sum_{j=1}^{i} a_j$.
The approximate prefix sums of a sequence $A$ can be computed in $O(n)$ work and $O(\log\log n)$ depth on the CRCW PRAM.
\end{lemma}
We can therefore compute approximate prefix sums over the elements for each set to approximately estimate its size in $O(|N(s)|)$ work and $O(\log \log n)$ depth.
Note that the approximate sizes estimated by \cref{lem:sizeest} are guaranteed to be over-approximations, i.e., for a set $s$, we have that $|N(s)| \leq \tilde{d}(s) \leq (1+\delta)|N(s)$ where $\tilde{d}(s)$ is its approximate size.

\begin{lemma}\label{lem:logn-approxsize}
Consider running Algorithm~\ref{alg:log-apx} where instead of using the exact size of every set $s \in S$, we use the approximate sizes, $\tilde{d}(s)$ where $|N(s)| \leq \tilde{d}(s) \leq (1+\delta)|N(s)|$. Then analogously to \cref{lem:logn-properties}, we have that (a) the residual size of each set in $D$ is at most $(1+\epsilon)(1+\delta)$ smaller than the maximum residual size of any set at that moment, and (b) for each newly covered element $t$, the expected number of sets in a batch that cover it at most $(1+4\epsilon)$.
\end{lemma}
\begin{proof}
Claim (a) follows since for $i = 0$ and the current value of $j$, the maximum set has size at most $(1+\epsilon)^{j+1}$, and the smallest set that can be placed in this bucket has size at least $(1+\epsilon)^{j} / (1+\delta)$. Claim (b) follows from \cref{lem:main-sampling}.
\end{proof}

An analogous argument to \cref{lemma:set-cover-log-approx} shows the following approximation guarantees on our implementation:
\begin{corollary}
Our implementation of \cref{alg:log-apx} in the CRCW PRAM computes an $(1+\epsilon)(1+\delta)(1+4\epsilon) \hdelta$-approximate (in expectation) solution to \setcover.
\end{corollary}

\paragraph*{Rebucketing Sets.}
Another question is how to deal with sets that drop below the size constraint, and should be rebucketed to a later round $j' < j$.
The number of such sets is at most $|S_j|$, and these sets must be rebucketed to rounds $[0, j)$ based on their new sizes.

We first provide more details on how the inner loop indexed by $i$ is implemented.
At the start of the round, we apply the Alias method (\cref{lem:aliasmethod}) to identify the step $i$ each set in $S_j$ gets sampled.
We then group the sets in $S_j$ based on their chosen steps using a parallel counting sort, similar to our implementation of Algorithm~\ref{alg:f-apx}.

At the start of each step, we estimate the size of each set that gets sampled in this step using approximate prefix-sums.
If a set in step $i$ no longer has sufficient size, we mark this set as needing rebucketing, but wait until the end of the round to rebucket it.
We mark the remaining sets in this step with a special flag to indicate that they are chosen in the set cover, and map over their incident elements to mark them all as covered.
At the end of a round, we collect all of the sets that were marked as needing rebucketing to obtain a set of sets $S'_j$ and rebucket them using a parallel counting sort.

\paragraph*{Parallel Cost in the CRCW PRAM.}
The initialization, which buckets the sets based on their initial degrees using parallel counting sort, can be implemented in $O(n)$ work and $O(\log n)$ depth. 

At the start of each round, we use the Alias method, which is called to partition the sets in $S_j$ into the $k$ steps with $k = O(\log f) = O(\log n)$.
We can use the sequential algorithm of Walker~\cite{walker1974new} or Vose~\cite{vose1991linear}, followed by a parallel counting sort to assign the sets in $S_j$ to the $k$ steps in $O(|S_j|)$ work and $O(\log n)$ depth.
The overall work of round $j$ is $O(|S_j| + \sum_{s \in S_j} |N(s)|)$ across all of its steps, since we process each set in $S_j$ in exactly one step, and in this step perform $O(|N(s)|)$ work on its elements.
Next, each step of a round uses $O(\log \log n)$ depth.
Thus, the overall work of a round is $O(|S_j|)$ and the overall depth is $O(\log n \log\log n)$.

One final issue is to account for the work of processing sets that are rebucketed. However, notice that any set that is rebucketed in a round $j$ initially had size at least $(1+\epsilon)^{j}$ at the time it was bucketed to round $j$, and subsequently fell below this size threshold.
Therefore, each time a set is rebucketed, its size must drop by at least a $(1+\epsilon)$ factor.
Let $P(s)$ denote the rounds that a set $s$ is processed in.
The largest value in $P(s)$ is $\lfloor \log_{1+\epsilon} |N(s)| \rfloor$ (i.e., the bucket corresponding to $s$'s initial size).
Summing over all rounds that a set $s$ is processed in gives a geometric sum $\sum_{j' \in P(s)} (1+\epsilon)^{j'}$, which evaluates to $O(|N(s)|)$ for any fixed constant $\epsilon$.

Therefore, the overall work of the algorithm is $O(n + m)$, and the overall depth of the algorithm is $O(\log^2 n \log\log n)$.

\paragraph*{Choice of Model and Other Parallel Models.}
We use the PRAM model to implement our algorithms to provide an accurate comparison with prior works on parallel \setcover, which all utilize the PRAM model.
However, most modern shared-memory parallel algorithms are designed in the Binary-Forking Model~\cite{blelloch2020optimal}, which is arguably a more realistic model that better captures how shared-memory algorithms are programmed and implemented today.
In particular, the binary-forking (BF) model more naturally captures nested-parallel programming, while getting rid of some unrealistic aspects of PRAMs, e.g., lock-step scheduling and the ability to perform arbitrary-way forking in constant depth.
Briefly, computation in the BF model can be viewed as a random-access machine (RAM) program with additional primitives to {\em fork} two sub-computations (themselves BF programs) and {\em join} back to the parent BF program once both sub-computations are finished.

Our PRAM implementations described in the previous sub-sections immediately give state-of-the-art BF implementations. 
In fact, some of the technical details about size estimation are completely eliminated since any BF computation that performs a parallel operation over $n$ objects requires $O(\log n)$ depth. Thus, we can afford to perform exact size computation in every step.
Implemented in the BF model, it is not hard to see that we require $O(\log^2 n)$ depth in BF for Algorithm~\ref{alg:f-apx}, and $O(\log^3 n)$ depth for Algorithm~\ref{alg:log-apx}.
Understanding whether one can achieve lower depth in the BF model for $f$ or $\hdelta$ approximation is an interesting direction for future work.

\section*{Acknowledgements}
We would like to thank Guy Blelloch and Mohsen Ghaffari for many insightful discussions during various stages of this project.
We are grateful to anonymous reviewers for their valuable feedback.
M.~Dinitz was in part supported by NSF awards CCF-1909111 and CCF-2228995.
S.~Mitrovi\'c was supported by the Google Research Scholar and NSF Faculty Early Career Development Program No.~2340048.

\bibliographystyle{alpha}
\bibliography{ref}

\end{document}